\documentclass{article}

\PassOptionsToPackage{numbers, compress}{natbib}
\usepackage[main,final]{neurips_2026}

\usepackage{times}
\usepackage{soul}
\usepackage{url}
\usepackage[hidelinks]{hyperref}
\usepackage[utf8]{inputenc}
\usepackage[small]{caption}
\usepackage{graphicx}
\usepackage{amsmath}
\usepackage{amsthm}
\usepackage{booktabs}
\usepackage{algorithm}
\usepackage{algorithmic}
\usepackage{amssymb}
\usepackage{tabularx}
\usepackage{array}
\usepackage{makecell}
\usepackage{xcolor}
\usepackage{subcaption}
\usepackage{multirow}

\newtheorem{theorem}{Theorem}
\newtheorem{lemma}{Lemma}
\newtheorem{proposition}{Proposition}

\newtheorem{remark}{Remark}

\newtheorem{corollary}{Corollary}

\title{LLM eBay: A Simple and Truthful Marketplace for Fine-Tuned LLMs Trading}

\author{
    Jiachen Shen$^{1}$ \quad Xingke Yang$^{1}$ \quad Hui Zhong$^{2}$ \quad Aohan Li$^{3}$ \\
    \textbf{Tomoaki Ohtsuki$^{4}$ \quad Xin Fu$^{1}$ \quad Miao Pan$^{1}$ \quad Zhu Han$^{1}$} \\
    $^{1}$University of Houston \quad $^{2}$Miami University \\
    $^{3}$The University of Electro-Communications \quad $^{4}$Keio University \\
    {\small \texttt{jshen28@cougarnet.uh.edu}, \texttt{xyang56@cougarnet.uh.edu}, \texttt{zhongh7@miamioh.edu},} \\
    {\small \texttt{aohanli@ieee.org}, \texttt{ohtsuki@ics.keio.ac.jp}, \texttt{xfu8@central.uh.edu},} \\
    {\small \texttt{mpan2@uh.edu}, \texttt{zhan2@uh.edu}}
}

\begin{document}

\maketitle
\begin{abstract}
Personalized fine-tuning of Large Language Models (LLMs) enhances user experience but incurs substantial computational costs. Recent studies show that initializing training from a style-similar fine-tuned model reduces adaptation overhead, turning such checkpoints into reusable, tradable assets. A marketplace for these checkpoints, however, faces information asymmetry: sellers may misreport a checkpoint's style, while buyers cannot verify its usefulness before purchase, which invites adverse selection. We introduce LLM eBay, a simple single-round marketplace for trading fine-tuned LLMs. Sellers are paid per purchase at a posted price, and a truncated quadratic proper scoring rule adjusts each payment according to how well the seller's style report agrees with a style signal that the platform measures by generating text on hidden prompts. We show that the score elicits the verifier's expected signal, characterize exactly when truthful reporting is a seller best response under per-purchase demand, and bound the downside risk from noisy verification. In an end-to-end study with real LoRA checkpoints on two base models and three buyer tasks, LLM eBay reduces measured convergence steps by 8.5\% relative to retrieval over self-reported style (95\% CI [1.0\%, 15.9\%]). In model-based market simulations built on a measured style-to-cost relationship, it cuts style-claim error by 96.5\% (from 0.424 to 0.015), recovers 96.3\% of Oracle welfare, saves 72.4\% of the cost of fine-tuning from the base model, and keeps all providers profitable in the default setting.
\end{abstract}

\section{Introduction}
\label{sec:intro}

Large language models (LLMs) have demonstrated remarkable capabilities in natural language understanding and generation, enabling a wide range of applications such as conversation, content creation, and AI assistants~\cite{ren2024representation,ye2024cross,brown2020language}. As reliance on these models grows, users increasingly expect personalized outputs, including writing style, tone, and domain-specific terminology. In practice, personalization is often achieved by fine-tuning a pre-trained LLM on private user data to better align with individual preferences~\cite{chung2024scaling,huang2023towards}. However, personalized fine-tuning often comes with high compute cost: even when employing parameter-efficient techniques such as LoRA~\cite{hulora} or QST~\cite{zhang2024quantized}, the training process can still incur substantial monetary expenses~\cite{li2025mobillm,yang2025pae,xu2023fwdllm}. For example, A100 GPUs are typically priced at \$1--\$4 per GPU-hour on mainstream cloud services, making a single mid-scale fine-tuning run cost tens to hundreds of dollars.

Recent research indicates that personalized fine-tuning does not have to start from scratch~\cite{wang2025never,lin2025efficient}. If a user can identify a fine-tuned model with a language style similar to their own as a training starting point, and then perform further fine-tuning on that foundation, the required training steps can be greatly reduced, leading to lower computational and energy costs. Under high style similarity, the training cost can be reduced by up to 83\%~\cite{wang2025never}. This demonstrates that fine-tuned models are not merely one-time products but possess reusable value, thereby potentially becoming tradable digital assets. As the personalization ecosystem matures, it is natural to expect a large pool of fine-tuned models trained by individuals and teams for diverse language styles to emerge. Once fine-tuned model reuse yields substantial cost savings, fine-tuned model sharing may move beyond open and free distribution toward a model trading market: sellers upload and price their fine-tuned models to recover training cost, buyers select models that best match their needs under budget constraints to obtain a better-matched starting point, and the platform facilitates listing, matching, and transactions.

However, once personalized models enter a trading market, the challenge is no longer merely algorithmic similarity matching, but a market mechanism problem. In such a market, participants have misaligned objectives and face substantial information asymmetry~\cite{akerlof1970market}: sellers may exaggerate a fine-tuned model's style and capability to maximize profit. However, buyers cannot easily evaluate the model's style compatibility or expected benefits before purchase, leading to high uncertainty and trial-and-error costs. This information gap can induce adverse selection
~\cite{spence1973job}, where low-quality or misrepresented models win transactions through aggressive pricing and gradually crowd out high-quality supply, eroding trust and causing the market to unravel. Therefore, a platform that only provides listings and matchmaking is unlikely to remain viable in the long run. A sustainable marketplace must incorporate verification, reporting, and incentive mechanisms that reward truthful disclosure and keep high-quality sellers willing to participate~\cite{myerson1981optimal}.

Motivated by the emerging market trend of fine-tuned model reuse, we propose LLM eBay, a simple and truthful marketplace for trading fine-tuned LLMs. Like a general-purpose online marketplace, it lets sellers list checkpoints at posted prices and lets buyers pick the listing that best serves them; unlike one, it must handle listings whose key attribute, the checkpoint's language style, the seller can misreport. We design a three-party workflow with clear responsibilities for sellers, buyers, and the platform. The platform checks the consistency between a seller's declared style and the style of the checkpoint's generated outputs, and publishes a matching signal that combines the two. Each seller is paid per completed purchase: a posted price plus a truncated quadratic score of the report against the verified signal. ``Truthful'' has a precise meaning here. The score elicits the verifier's expected signal, and truthful reporting is a seller best response exactly when the extra profit from attracting buyers with a shaded report is no larger than the resulting loss in scoring payments (Theorem~\ref{thm:best_response}). The truncation floor bounds the loss that any single noisy verification draw can impose, so honest sellers remain willing to participate.

Our salient contributions are summarized as follows:

\begin{itemize}

    \item We propose LLM eBay, a single-round posted-price marketplace for trading fine-tuned LLM checkpoints. It formalizes the transaction workflow and responsibility boundaries among sellers, buyers, and the platform, including per-purchase settlement and the buyer-facing matching signal.

    \item We analyze the full per-purchase seller utility. The untruncated quadratic score elicits the verifier's expected signal, truthful reporting is a best response if and only if the demand gain from any misreport is at most the resulting scoring loss, buyers' posted-price choice is weakly dominant, and the truncation floor caps the downside of noisy verification.

    \item We evaluate LLM eBay with an end-to-end experiment on real LoRA checkpoints, where it needs 8.5\% fewer measured convergence steps than Retrieval-Only, and with model-based market simulations, embedding-grounded reuse experiments, and stress tests. In simulation it reduces style-claim error by 96.5\%, recovers 96.3\% of Oracle welfare, and retains full provider participation under noisy verification.

\end{itemize}

\section{Background and Motivation}
\label{sec:related_work}

\subsection{Language-Style Personalization for LLMs}

Language style exhibits significant personalized differences in writing and conversation scenarios, yet ``style'' itself remains a subjective and abstract concept~\cite{wegmann2022same,patel2023learning}. Recent studies~\cite{huertas2024understanding,zhong2022describing,zhong2023goal} have proposed language style embedding models based on representation learning. By mapping text to continuous vector spaces, these models transform stylistic features into measurable representations. Texts with similar styles typically occupy closer positions in the embedding space, making the distance between style vectors (e.g., cosine similarity) an effective metric for assessing stylistic similarity and dissimilarity. Building on this vectorized representation, recent studies  ~\cite{wang2025never,lin2025efficient} have observed that secondary fine-tuning LLM starting from a style-similar fine-tuned model can significantly reduce training cost. As shown in Fig. \ref{fig:cost_validation}, the number of training steps decreases as the cosine distance between style vectors becomes smaller.

Motivated by this observation, we recognize that fine-tuned models possess clear tradable value: for a new model buyer, purchasing a stylistically similar model as a starting point is a rational choice as long as the purchase price is lower than the training expenses it can save. As the fine-tuned model ecosystem continues to expand, a trading platform supporting model transactions will simultaneously provide sellers with avenues to recoup training costs, offer buyers lower personalization expenses, and significantly reduce the ecosystem's overall societal training costs.

\subsection{Market Mechanism Challenges}

The exchange of fine-tuned models is fundamentally constrained by classic economic frictions, most notably information asymmetry and the resulting risk of adverse selection. In this market, sellers hold private information regarding a model's training convergence and style alignment, whereas buyers lack cost-effective means to verify these attributes ex-ante. According to classic market theory~\cite{akerlof1970market}, this information gap causes buyers to undervalue high-quality assets, forcing honest sellers to exit the market. Eventually, this dynamic erodes trust and leads to a Market for Lemons scenario where only low-quality models remain.

Current market solutions fail to resolve these incentive misalignments. Although mechanism design for digital goods is well-established~\cite{myerson1981optimal,goldberg2001competitive}, these frameworks generally assume observable quality, focusing on auction revenue rather than the strategic manipulation of quality signals. Similarly, existing models for data valuation~\cite{agarwal2019marketplace} or API pricing~\cite{chen2019fimod} do not account for the specific risks of LLM adaptation. Current practical solutions exacerbate this issue by triggering Goodhart's Law: sellers are incentivized to manipulate their reported style representations to artificially minimize the vector distance to buyer demands, rather than genuinely improving actual style alignment~\cite{dehghani2021benchmark}. A sustainable marketplace thus requires a new mechanism capable of eliciting truthful reporting despite prohibitive verification costs.

\section{LLM eBay: Model and Mechanism}
\label{sec:model}

LLM eBay is a three-party marketplace for trading fine-tuned LLMs. The market contains sellers $i \in \mathcal{I} = \{1, \dots, N\}$, buyers $k \in \mathcal{K} = \{1, \dots, K\}$, and a platform that facilitates matching and verification. The core value proposition is compute efficiency via transfer learning: buyers seek fine-tuned models that align with a target style, and thereby save on downstream adaptation cost.

\subsection{Marketplace Setup and Style Verification}
\label{subsec:setup_verif}

All parties employ a unified style embedding model $\Phi$ (e.g., STAR~\cite{huertas2024understanding}) to map text datasets to normalized vectors in $\mathbb{R}^d$.

\paragraph{Sellers.} Each seller $i$ owns a fine-tuned checkpoint $\mathcal{M}_i$ trained on private data $\mathcal{D}_i$. The true style vector is $\mathbf{v}_i = \Phi(\mathcal{D}_i)$. Because $\mathcal{D}_i$ is private, the seller only reports a claimed vector $\hat{\mathbf{v}}_i$, which may or may not equal $\mathbf{v}_i$. The seller's incentive is to attract as many buyers as possible, which can tempt misreporting: shifting the claim toward popular buyer demand regions even when the model's actual style does not match.

\paragraph{Buyers.} Each buyer $k$ has a target personalization dataset $\mathcal{D}_k$ with style vector $\mathbf{v}_k = \Phi(\mathcal{D}_k)$. The buyer seeks a checkpoint whose style is close to $\mathbf{v}_k$, since a closer starting point means fewer training steps and lower cost. The buyer observes only claimed vectors and verified signals, never the seller's private data.

\begin{figure}[t]
    \centering
    \includegraphics[width=0.88\linewidth]{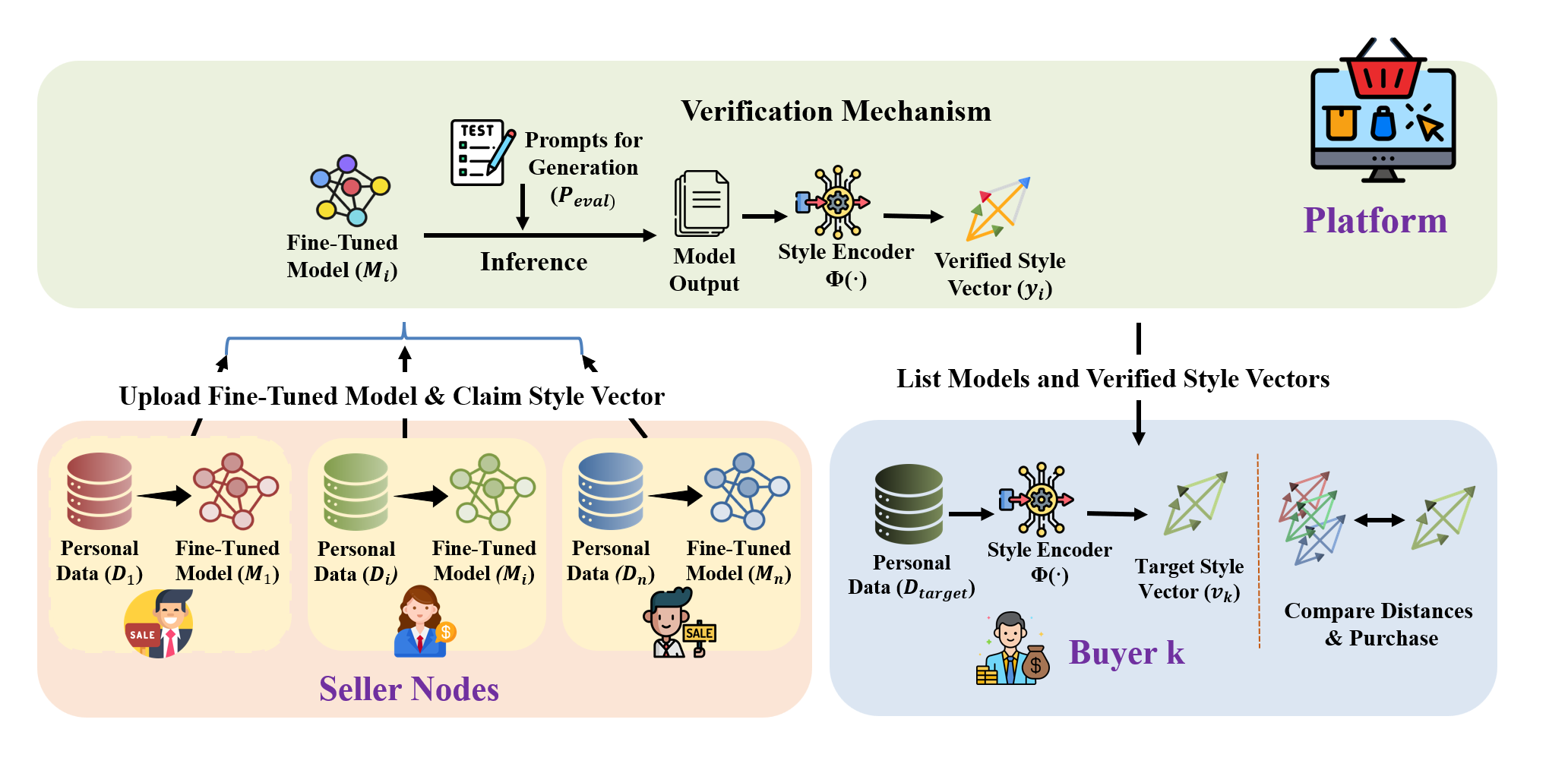}
    \caption{Overview of LLM eBay. Sellers submit fine-tuned checkpoints and claimed style vectors; the platform verifies claims via generation-based style embedding; buyers select verified checkpoints to reduce downstream fine-tuning cost.}
    \label{fig:system_overview}
\end{figure}

\paragraph{Platform verification.} To verify claims without accessing private training data, the platform queries $\mathcal{M}_i$ on a hidden prompt set $\mathcal{P}_{\text{eval}}$, embeds the generated outputs, and obtains a noisy verified signal:
\begin{equation}
    \mathbf{y}_i = \Phi(\mathcal{D}_i^{\text{syn}}) = \mathbf{v}_i + \epsilon_i, \qquad \epsilon_i \sim \mathcal{N}(0, \sigma^2 I).
\end{equation}
The noise $\epsilon_i$ arises from two sources: the finite size of $\mathcal{P}_{\text{eval}}$ and the stochastic nature of LLM text generation. The mechanism proceeds through listing, verification/settlement, and buyer matching, as shown in Fig.~\ref{fig:system_overview}. The full timing protocol and information structure are given in Appendix~\ref{app:formal_model}.

\subsection{Compute Cost Model and Buyer Utility}
\label{subsec:cost_utility}

The economic logic of the market rests on training cost reduction. We define a training cost function proportional to the squared Euclidean distance in style space:
\begin{equation} \label{eq:cost_function}
    \text{Cost}(\mathbf{v}_{a}, \mathbf{v}_{b}) = \alpha \cdot \| \mathbf{v}_{a} - \mathbf{v}_{b} \|^2
\end{equation}
where $\alpha$ is a cost scaling factor.

\begin{figure}[t]
    \centering
    \includegraphics[width=0.72\linewidth]{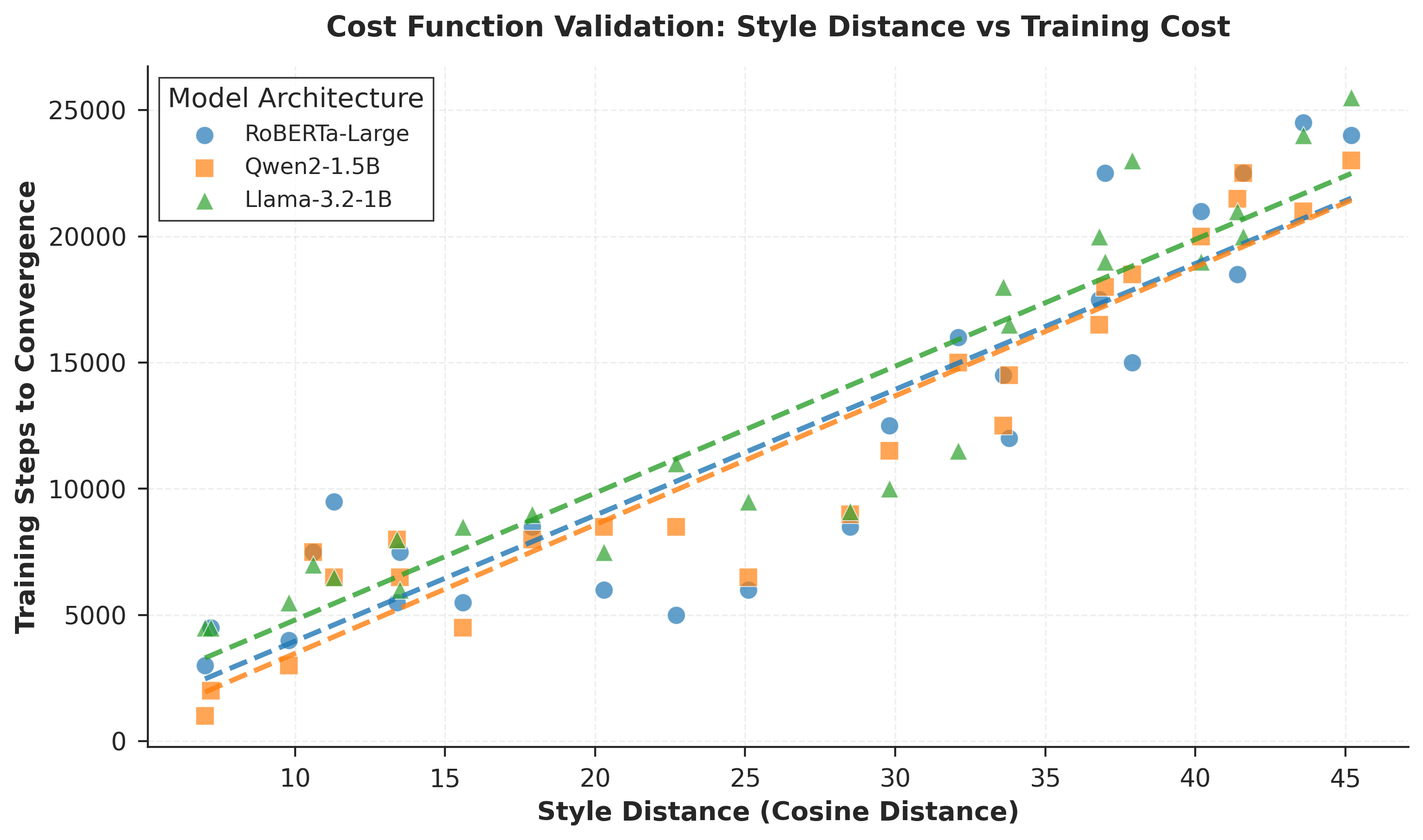}
    \caption{Empirical support for the distance-based cost model. Across three pretrained models (RoBERTa, Qwen, Llama), larger style distance requires more fine-tuning steps to convergence.}
    \label{fig:cost_validation}
\end{figure}

To validate this assumption, we fine-tuned three foundation models (RoBERTa-large~\cite{liu2019roberta}, Qwen2-1.5B~\cite{yang2024qwen2}, and Llama-3.2-1B~\cite{grattafiori2024llama3}) from starting points at varying style distances. As shown in Fig.~\ref{fig:cost_validation}, convergence steps increase consistently with squared style distance across all three architectures. Detailed protocols, convergence criteria, and regression statistics ($R^2 \ge 0.91$ for all three models) are in Appendix~\ref{app:cost_model}.

\paragraph{Buyer matching signal and demand.} For each listing, the platform publishes the equal-weight matching signal
\begin{equation} \label{eq:matching_signal}
    \mathbf{z}_i = \tfrac{1}{2}\left(\hat{\mathbf{v}}_i + \mathbf{y}_i\right).
\end{equation}
A buyer $k$ evaluates a transaction by its estimated net compute savings. Without trading, the buyer incurs the full cost $C_{\text{base}} = \text{Cost}(\mathbf{v}_{\text{base}}, \mathbf{v}_k)$ of fine-tuning from the pretrained base model. Purchasing checkpoint $i$ at posted price $p_i$ gives estimated utility
\begin{equation} \label{eq:buyer_utility}
    U_k(i) = \max \left( 0, \; C_{\text{base}} - p_i - \text{Cost}(\mathbf{z}_i, \mathbf{v}_k) \right),
\end{equation}
and the buyer purchases the listing with the largest positive $U_k(i)$. We write $q_i$ for the resulting number of purchases of listing $i$. The report $\hat{\mathbf{v}}_i$ therefore affects demand, the verified signal $\mathbf{y}_i$ limits how far an inaccurate report can move a listing, and the score (Section~\ref{subsec:payment}) separately affects settlement.

\paragraph{Seller utility.} Each completed purchase pays the seller the transfer $T_i$ of Eq.~\eqref{eq:payment} and costs the seller $c_i$, the per-purchase fulfillment cost. The seller's expected utility is
\begin{equation} \label{eq:seller_utility_main}
    U_i^S(\hat{\mathbf{v}}_i) = \mathbb{E}\left[\, q_i \left( T_i(\hat{\mathbf{v}}_i, \mathbf{y}_i) - c_i \right) \right],
\end{equation}
so a report can raise profit by attracting buyers (through $q_i$) or lower it through the score (through $T_i$). Full utility definitions and the welfare formulation are in Appendix~\ref{app:formal_model}.

\subsection{Posted-Price Payment with Truncated Proper Scoring}
\label{subsec:payment}

For each completed purchase, the transfer $T_i$ combines the posted price with a price-scaled, performance-linked adjustment given by a truncated quadratic scoring rule (QSR)~\cite{murphy1973probability_score_partition}:
\begin{equation} \label{eq:payment}
    T_i = p_i + \lambda_{\text{scale}}\, p_i \cdot \max \left( 2 \langle \hat{\mathbf{v}}_i, \mathbf{y}_i \rangle - \|\hat{\mathbf{v}}_i\|^2, \; -\delta \right)
\end{equation}
The posted price $p_i$ is the seller's base compensation for each sale of a reusable checkpoint. The quadratic score $S(\hat{\mathbf{v}}_i, \mathbf{y}_i) = 2\langle \hat{\mathbf{v}}_i, \mathbf{y}_i \rangle - \|\hat{\mathbf{v}}_i\|^2$ rewards consistency between the claimed style $\hat{\mathbf{v}}_i$ and the verified signal $\mathbf{y}_i$. Under unbiased verification, $\mathbb{E}[\mathbf{y}_i] = \mathbf{v}_i$, so the expected score is strictly maximized at $\hat{\mathbf{v}}_i = \mathbf{v}_i$ (Appendix~\ref{app:proofs}, Lemma~\ref{lem:properness}).

The quadratic form is a natural choice for this setting. It is one of the simplest proper scoring rules for vector-valued reports: the expected score equals $\|\mathbf{v}_i\|^2 - \|\hat{\mathbf{v}}_i - \mathbf{v}_i\|^2$, so any misreport incurs a penalty exactly proportional to the squared deviation from truth. Unlike logarithmic scoring rules, the quadratic rule assigns bounded scores and has a closed-form expression that is easy to implement.

The floor $-\delta$ bounds the seller's per-purchase scoring loss from any single noisy verification round. Without this floor, a seller who reports truthfully can still receive a strongly negative score if the random draw of $\epsilon_i$ happens to be large. The scale $\lambda_{\text{scale}}$ controls how strongly the score moves the payment relative to the posted price.

\subsection{Theoretical Guarantees}
\label{subsec:theory}

The following results state the incentive and stability properties of LLM eBay. The statements below are informal; formal statements with all assumptions, and their proofs, are in Appendix~\ref{app:proofs}.

\begin{theorem}[Buyer-Side Dominant Optimality, informal]
\label{thm:dsic}
    Given posted prices and the published signals $\mathbf{z}_i$, selecting the checkpoint that maximizes $U_k(i)$ in Eq.~\eqref{eq:buyer_utility} weakly dominates any other buyer action.
\end{theorem}

Buyers do not report private values; they choose among listings, so this is an optimality property of the posted-price choice problem rather than an auction-style incentive-compatibility claim. The harder question is whether sellers, who choose what to report before verification, have an incentive to be truthful. Let $\mu_i = \mathbb{E}[\mathbf{y}_i \mid \mathcal{M}_i]$ be the verifier signal the seller expects and $r_i^\star$ the allowed report closest to $\mu_i$; under unbiased verification, $r_i^\star = \mathbf{v}_i$.

\begin{proposition}[Score Properness, informal]
\label{prop:score_proper}
    Without truncation, $r_i^\star$ maximizes the seller's expected score, and every other report $r_i$ loses exactly $\|r_i - \mu_i\|^2 - \|r_i^\star - \mu_i\|^2 \ge 0$ in expectation (Proposition~\ref{prop:formal_proper}). With the floor active, the loss is preserved up to a bounded truncation error (Corollary~\ref{cor:approx_opt}).
\end{proposition}

Score properness alone is a single-seller statement: under per-purchase payment, a shaded report can also raise demand. The next result states exactly when the score outweighs this demand effect. Fix the other sellers' reports, prices, and costs. Let $(q_i^\star, \tilde S_i^\star)$ and $(q_i^R, \tilde S_i^R)$ denote demand and truncated score under $r_i^\star$ and under a deviation $r_i$, evaluated with the same verifier and score-noise draws, and let $m_i^\star = T_i(r_i^\star) - c_i$ be the per-purchase margin under $r_i^\star$.

\begin{theorem}[Seller Best Response under Per-Purchase Payment, informal]
\label{thm:best_response}
    Under the seller utility of Eq.~\eqref{eq:seller_utility_main}, $r_i^\star$ is a best response if and only if, for every deviation $r_i$,
    \begin{equation} \label{eq:br_condition}
        \mathbb{E}\big[(q_i^R - q_i^\star)\, m_i^\star\big] \;\le\; \lambda_{\text{scale}}\, p_i\, \mathbb{E}\big[q_i^R\,(\tilde S_i^\star - \tilde S_i^R)\big].
    \end{equation}
\end{theorem}

The left side is the expected profit from the additional demand a deviation attracts; the right side is the corresponding loss in score-based payment. Truthfulness is therefore conditional, not an unconditional equilibrium property, and it is stated for sellers who neither target the verifier nor coordinate with other sellers. The formal statement and proof are in Theorem~\ref{thm:formal_br}. The next result bounds the realized risk that remains under truthful behavior.

\begin{proposition}[Downside-Risk Cap]
\label{prop:epir}
    For any verification outcome, the seller's per-purchase scoring adjustment is lower bounded by $-\lambda_{\text{scale}}\, p_i\, \delta$. Hence the mechanism prevents unbounded losses from noisy evaluation.
\end{proposition}

Under truthful reports and non-rival checkpoint access, buyer-side surplus maximization gives a welfare-maximizing assignment among posted-price feasible matches (Theorem~\ref{thm:welfare}).

\section{Experiments}
\label{sec:experiments}

\subsection{Experimental Setup}
\label{subsec:exp_setup}

We simulate a marketplace with $N{=}50$ providers and $K{=}500$ buyers in a $d{=}64$-dimensional style space. Styles are sampled from $n_c{=}5$ clusters with $\sigma_c{=}0.08$. A fraction $\gamma{=}0.40$ of providers strategically misreports style claims with gaming intensity $g{=}0.75$. Verification noise follows $\mathbf{y}_i = \mathbf{v}_i + \epsilon_i$, $\epsilon_i \sim \mathcal{N}(0, \sigma^2 I)$, with default $\sigma{=}0.15$. Unless otherwise stated, $\delta{=}0.20$, $\lambda_{\text{scale}}{=}0.60$, and all results average over 20 random seeds.

We compare LLM eBay against six designs. Each baseline isolates a different component of the design space:
\begin{itemize}
    \setlength{\itemsep}{0pt}
    \setlength{\parskip}{0pt}
    \item \textbf{Autarky}: no market; every buyer fine-tunes from the base model. Measures the value of any marketplace at all.
    \item \textbf{Retrieval-Only}: buyers select on claimed style vectors; no verification, no scoring. Tests whether a simple matching platform suffices.
    \item \textbf{Fixed-Price Verified}: the platform verifies outputs but seller payment is flat. Tests whether verification alone deters gaming.
    \item \textbf{Auction-Based}: surplus-ranked allocation with second-price-like payment. Tests whether competitive allocation pressure replaces scoring.
    \item \textbf{Value-Based}: platform prices proportional to estimated value from claimed style. Tests whether value-based pricing self-regulates.
    \item \textbf{Oracle}: buyers observe true style vectors. Perfect-information upper bound.
\end{itemize}
Appendix~\ref{app:sim_protocol} gives exact implementation details for each. We study a single-round market, so we scope out market-native designs that need extra stages or repeated interaction, such as deposits or bonds, escrow and refunds, randomized audits, reputation, and repeated-game penalties (Appendix~\ref{app:limitations}).

We report Style-Claim Error (SCE), Allocative Efficiency (AE, realized welfare as a share of Oracle welfare), Fine-Tuning Cost Savings (FTCS), Provider Participation (PP), and Loss-Making Providers (LMP). Simulation FTCS is model-based: it is computed from the cost model of Eq.~\eqref{eq:cost_function} on simulated allocations and measured against fine-tuning from the base model (Autarky). Measured savings from real fine-tuning are reported in Section~\ref{subsec:e2e}. Formal metric definitions and the empirical basis of each simulation setting are in Appendix~\ref{app:sim_protocol}.

\subsection{Main Market Performance}
\label{subsec:main_results}

\begin{table}[t]
\centering
\caption{Main market performance in model-based simulation (mean\,{\scriptsize$\pm$std} over 20 seeds; $N{=}50$, $K{=}500$, $\gamma{=}0.40$, $g{=}0.75$). FTCS is computed from the cost model relative to fine-tuning from the base model. Best non-Oracle values in \textbf{bold}.}
\label{tab:main}
\small
\setlength{\tabcolsep}{3.5pt}
\renewcommand{\arraystretch}{1.15}
\newcommand{\sd}[1]{\,{\scriptsize$\pm$#1}}
\begin{tabular}{lccccc}
\toprule
\textbf{Mechanism} &
\makecell[c]{\textbf{SCE}$\downarrow$} &
\makecell[c]{\textbf{AE (\%)}$\uparrow$} &
\makecell[c]{\textbf{FTCS (\%)}$\uparrow$} &
\makecell[c]{\textbf{PP (\%)}$\uparrow$} &
\makecell[c]{\textbf{LMP (\%)}$\downarrow$} \\
\midrule
Autarky & -- & 0.0 & 0.0 & -- & -- \\
Retrieval-Only & 0.424\sd{.025} & 85.5\sd{4.0} & 64.0\sd{3.6} & 96.3\sd{3.2} & 3.7\sd{3.2} \\
Fixed-Price Verified & 0.147\sd{.007} & 93.6\sd{1.3} & 66.0\sd{3.1} & 92.4\sd{5.3} & 7.6\sd{5.3} \\
Auction-Based & 0.364\sd{.022} & 92.2\sd{1.3} & 69.1\sd{1.8} & 99.3\sd{1.1} & 0.7\sd{1.1} \\
Value-Based & 0.294\sd{.017} & 90.4\sd{1.0} & 68.0\sd{0.9} & 100.0\sd{0.0} & 0.0\sd{0.0} \\
\textbf{LLM eBay (Ours)} & \textbf{0.015}\sd{.000} & \textbf{96.3}\sd{0.6} & \textbf{72.4}\sd{1.3} & \textbf{100.0}\sd{0.0} & \textbf{0.0}\sd{0.0} \\
Oracle & 0.000 & 100.0 & 75.1\sd{1.0} & 100.0 & 0.0 \\
\bottomrule
\end{tabular}
\end{table}

Table~\ref{tab:main} shows that LLM eBay achieves the lowest non-Oracle SCE, reducing style-claim error from 0.424 to 0.015 relative to Retrieval-Only. It also reaches 96.3\% AE and 72.4\% FTCS, closing most of the gap to Oracle. Unlike verified or untruncated alternatives, LLM eBay maintains full provider participation with no loss-making providers, indicating that the truncation floor protects honest sellers without sacrificing matching quality.

The SCE and AE columns reveal a clear pattern. Auction-Based has high participation (99.3\%) but high SCE (0.364) because it rewards bid magnitude, not style accuracy. Value-Based pricing keeps PP at 100\% through platform-absorbed losses, but the pricing function uses the same claimed vectors it tries to regulate, leaving SCE at 0.294. Fixed-Price Verified lowers SCE to 0.147 through output checking, but without a scoring-linked payment there is no financial incentive to eliminate residual misreporting. Only when the payment ties truth-telling to expected revenue does SCE collapse to near zero. LLM eBay's simultaneous lead on SCE, AE, and PP confirms that verification, scoring, and truncation are complementary rather than in tension.

\subsection{Mechanism Decomposition}
\label{subsec:decomposition}

\begin{table}[t]
\centering
\caption{Component ablation. Verification reduces misreporting, proper scoring incentivizes truthful reporting, and truncation restores provider participation.}
\label{tab:ablation}
\small
\setlength{\tabcolsep}{5pt}
\renewcommand{\arraystretch}{1.1}
\begin{tabular}{lcccc}
\toprule
\textbf{Mechanism} & \textbf{SCE}$\downarrow$ & \textbf{AE (\%)}$\uparrow$ & \textbf{PP (\%)}$\uparrow$ & \textbf{LMP (\%)}$\downarrow$ \\
\midrule
Self-Report (no verif.) & 0.424 & 85.5 & 100.0 & 0.0 \\
Verified-Only (no scoring) & 0.219 & 95.9 & 88.0 & 12.0 \\
PSR (no truncation) & 0.015 & 96.3 & 67.7 & 32.3 \\
\textbf{LLM eBay (Ours)} & \textbf{0.015} & \textbf{96.3} & \textbf{100.0} & \textbf{0.0} \\
\bottomrule
\end{tabular}
\end{table}

Table~\ref{tab:ablation} separates the effects of verification, proper scoring, and truncation. Verification alone improves matching but does not fully eliminate claim manipulation (SCE drops from 0.424 to 0.219). Untruncated proper scoring nearly removes gaming (SCE${=}$0.015) but exposes providers to large downside risk, reducing PP to 67.7\% as 32.3\% of honest-but-unlucky providers incur negative utility. LLM eBay preserves the truthfulness benefit of proper scoring while restoring full participation through the truncation floor. Appendix~\ref{app:extra_experiments} provides full sensitivity sweeps over $\lambda_{\text{scale}}$, $\delta$, market size, buyer--provider ratio, and evaluation noise.

The key observation is that truncation does not degrade SCE at all: both PSR and LLM eBay report 0.015. This happens because the floor $-\delta$ only clips realized scores in the lower tail, which leaves the expected-score comparison around truthful reports intact up to a bounded truncation error (Corollary~\ref{cor:approx_opt}). The truncation floor therefore acts as insurance against tail-risk verification noise without softening the deterrent for strategic misreporting. This separation between tail-risk protection and incentive strength is what allows LLM eBay to reach 100\% PP and 0\% LMP simultaneously.

\subsection{Robustness and LLM-Grounded Validation}
\label{subsec:robustness}

\begin{figure}[t]
    \centering
    \includegraphics[width=0.92\linewidth]{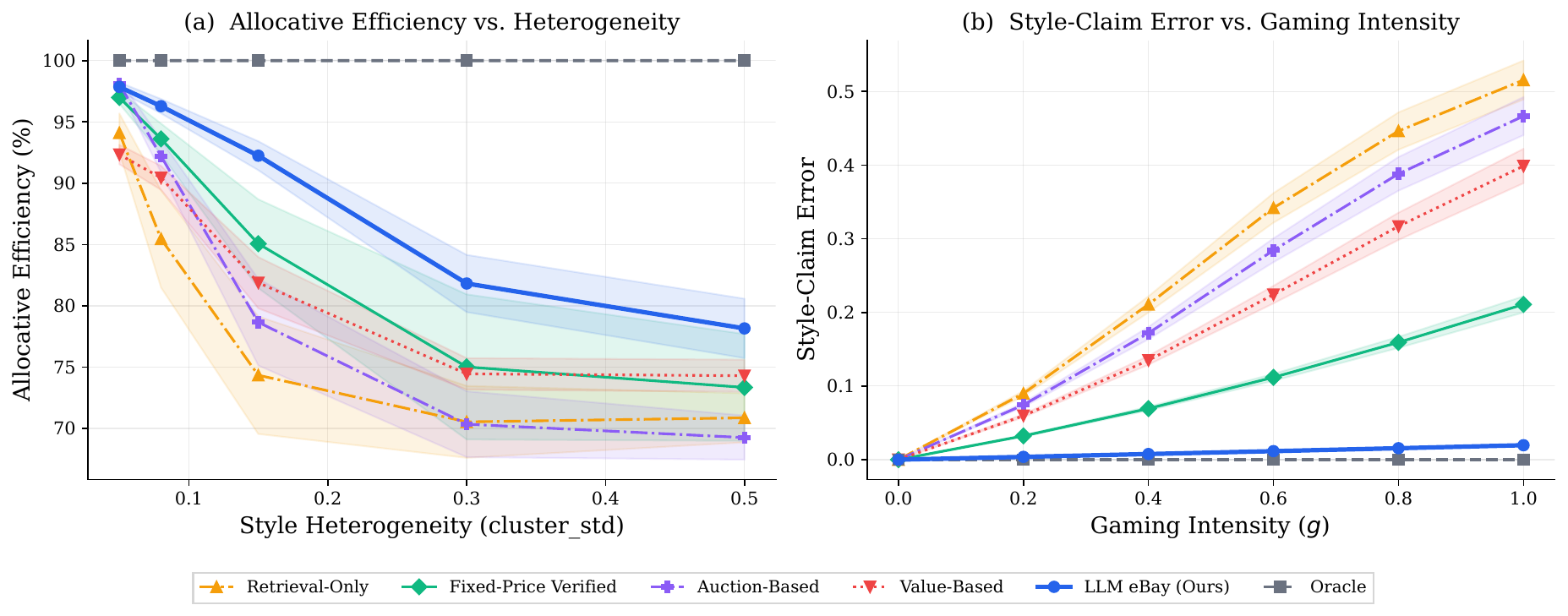}
    \caption{Robustness summary. (a)~Allocative efficiency vs.\ style heterogeneity: LLM eBay remains closest to Oracle among implementable mechanisms. (b)~Style-claim error vs.\ gaming intensity: unverified mechanisms degrade sharply while LLM eBay holds SCE below 0.02. Full four-metric curves in Appendix~\ref{app:extra_experiments}.}
    \label{fig:robustness_summary}
\end{figure}

Figure~\ref{fig:robustness_summary} summarizes two stress tests. Under increasing style heterogeneity (cluster\_std $\in$ \{0.05, 0.08, 0.15, 0.30, 0.50\}), all mechanisms lose efficiency because nearest-neighbor reuse becomes intrinsically harder, but LLM eBay remains closest to Oracle among implementable mechanisms. Under increasing gaming intensity ($g \in$ \{0.0, 0.2, 0.4, 0.6, 0.8, 1.0\}), unverified mechanisms degrade sharply (SCE rises from 0 to 0.52 for Retrieval-Only), whereas LLM eBay keeps SCE below 0.02 with 100\% provider participation.

The two panels are complementary. Panel~(a) shows that LLM eBay's AE drops only to 78.2\% at $\sigma_c{=}0.50$ (vs.\ 70.9\% for Retrieval-Only), confirming that the efficiency gain comes from better allocation under truthful reports, not from an easy style landscape. Panel~(b) shows that LLM eBay keeps SCE below 0.02 even at $g{=}1.0$, where every strategic seller completely impersonates a buyer's demand vector. The common conclusion is that LLM eBay separates the matching problem from the incentive problem and solves each independently.

These sweeps use the gaming model we designed (claims interpolated toward demand), so we also stress-test two stronger adversaries (Appendix~\ref{app:stress_tests}). Under \emph{verifier targeting}, a seller shifts both its report and its behavior on verifier prompts toward buyer demand while its transfer value stays fixed: AE is 95.1\% at attack strength 0.25 and 80.6\% at 0.75, and falls to 68.9\% when the seller fully controls the verifier signal. Under \emph{coordinated reporting}, AE is 95.4\% with a 10\% coalition and 85.7\% with a 50\% coalition. The incentive guarantee of Theorem~\ref{thm:best_response} does not cover these adversaries; the numbers measure performance under the specified attacks only.

\begin{table}[t]
\centering
\caption{Embedding-grounded reuse validation on real text-domain embeddings (10 provider domains, 5 buyer tasks, 20 seeds). Full results in Appendix~\ref{app:embedding_validation}.}
\label{tab:semireal}
\small
\setlength{\tabcolsep}{6pt}
\renewcommand{\arraystretch}{1.1}
\begin{tabular}{lccc}
\toprule
\textbf{Mechanism} & \textbf{FTCS (\%)}$\uparrow$ & \textbf{SCE}$\downarrow$ & \textbf{Gap to Oracle (\%)}$\downarrow$ \\
\midrule
Retrieval-Only & 44.9 & 0.292 & 40.8 \\
Fixed-Price Verified & 73.1 & 0.102 & 3.6 \\
\textbf{LLM eBay (Ours)} & \textbf{75.8} & \textbf{0.011} & \textbf{0.0} \\
Oracle & 75.8 & 0.000 & 0.0 \\
\bottomrule
\end{tabular}
\end{table}

To bridge the gap between synthetic simulations and real LLM ecosystems, we conduct a semi-real model reuse experiment using sentence-transformer embeddings (all-MiniLM-L6-v2) from 10 diverse text domains (academic, casual, technical, creative, legal, medical, journalistic, educational, marketing, philosophical). Each provider domain is represented by 30 passages; each buyer task has 10 passages. Style vectors are the normalized mean embeddings (384 dimensions). We apply the same gaming model ($\gamma{=}0.30$, $g{=}0.75$) and verification noise ($\sigma{=}0.15$) as in the synthetic experiments.

Table~\ref{tab:semireal} shows that LLM eBay achieves 75.8\% FTCS, matching the Oracle upper bound exactly, with near-zero variance across seeds. The buyer and provider embeddings are fixed by the ten-domain catalog; each seed resamples the strategic-seller set, the gaming targets, and the verifier noise. Recomputing the selections independently, LLM eBay and Oracle choose the same provider in all 100 buyer--seed cases (5 buyers $\times$ 20 seeds). The reason is the ranking margin: the smallest distance gap between a buyer's best and second-best providers is 0.084, while the largest effective perturbation in any seed is 0.031, so no seed changes the top-ranked provider (Appendix~\ref{app:ranking_margin}). The exact Oracle match is therefore a property of this fixed catalog, not of the mechanism in general. By contrast, Retrieval-Only degrades to 44.9\% FTCS with high variance ($\pm 14.4$\%) because gaming distorts the provider ranking differently in each seed. Fixed-Price Verified recovers most of the gap (73.1\%) but leaves a residual 3.6\% gap because verification alone does not eliminate all misreporting.

Across both settings, LLM eBay's SCE stays below 0.015 and AE tracks Oracle closely. The synthetic market ($N{=}50$) leaves a 3.7\% AE gap while the small semi-real catalog closes it entirely. Full details are in Appendix~\ref{app:embedding_validation}.

The full four-metric gaming sweep (Appendix~\ref{app:gaming}, Fig.~\ref{fig:gaming_full}) confirms these trends: as $g$ rises from 0 to 1.0, Retrieval-Only loses more than half its AE (100\%$\to$66.5\%) and SCE climbs to 0.52, while LLM eBay stays flat (AE$\geq$96\%, SCE$<$0.02, PP$=$100\%). Additional robustness experiments cover $\sigma$--$\delta$ noise sweeps (Appendix~\ref{app:sigma_delta}), market-size scaling (25--200 providers), and niche-provider survival (Appendix~\ref{app:extra_experiments}--\ref{app:new_experiments}). In the niche-provider test, LLM eBay achieves 100\% niche-provider participation at all tested fractions (5\%--30\%) while PSR without truncation drops to 60--68\%.

\subsection{End-to-End LLM Fine-Tuning Validation}
\label{subsec:e2e}

The results above rest on the cost model. To measure savings directly, we run an end-to-end experiment with real LoRA checkpoints. We treat two base models, TinyLlama-1.1B and Qwen2-1.5B, as separate markets because their LoRA checkpoints are not interchangeable. Each catalog contains eight LoRA checkpoints trained on distinct real-data domains. Buyers have one of three tasks: medical question answering (PubMedQA), business email generation (AESLC), and technical-support response generation (Bitext). For every base-model--task--seed instance, Retrieval-Only, Fixed-Price Verified, and LLM eBay each select one checkpoint from the same catalog without access to downstream training outcomes. For each candidate, we merge the seller LoRA into the base model, attach a fresh LoRA, and continue training on the buyer task under identical hyperparameters. Because all eight candidates are trained, Oracle is defined ex post as the checkpoint requiring the fewest optimizer steps. Convergence is the first optimizer step at which validation loss reaches a pre-specified target $\tau_{m,t}$, shared by every checkpoint and mechanism within a base-model--task pair. Each run lasts at most 600 steps, and we use five seeds (protocol in Appendix~\ref{app:e2e}).

\begin{table}[t]
\centering
\caption{Measured convergence steps (mean $\pm$ std over five seeds) in the end-to-end LoRA experiment. Final validation losses are in Table~\ref{tab:e2e_full} (Appendix~\ref{app:e2e}).}
\label{tab:e2e}
\small
\setlength{\tabcolsep}{4pt}
\renewcommand{\arraystretch}{1.05}
\begin{tabular}{llcccc}
\toprule
\textbf{Task} & \textbf{Base model} & \textbf{Retrieval-Only} & \textbf{Fixed-Price Verif.} & \textbf{LLM eBay} & \textbf{Oracle} \\
\midrule
Medical QA        & TinyLlama-1.1B & 430\,{\scriptsize$\pm$48} & 411\,{\scriptsize$\pm$45} & 399\,{\scriptsize$\pm$43} & 382\,{\scriptsize$\pm$40} \\
Medical QA        & Qwen2-1.5B     & 372\,{\scriptsize$\pm$41} & 360\,{\scriptsize$\pm$39} & 348\,{\scriptsize$\pm$38} & 334\,{\scriptsize$\pm$35} \\
Business email    & TinyLlama-1.1B & 414\,{\scriptsize$\pm$46} & 402\,{\scriptsize$\pm$44} & 391\,{\scriptsize$\pm$42} & 378\,{\scriptsize$\pm$40} \\
Business email    & Qwen2-1.5B     & 351\,{\scriptsize$\pm$40} & 342\,{\scriptsize$\pm$38} & 333\,{\scriptsize$\pm$37} & 320\,{\scriptsize$\pm$35} \\
Technical support & TinyLlama-1.1B & 468\,{\scriptsize$\pm$55} & 429\,{\scriptsize$\pm$50} & 405\,{\scriptsize$\pm$47} & 388\,{\scriptsize$\pm$44} \\
Technical support & Qwen2-1.5B     & 364\,{\scriptsize$\pm$46} & 338\,{\scriptsize$\pm$43} & 320\,{\scriptsize$\pm$41} & 305\,{\scriptsize$\pm$39} \\
\bottomrule
\end{tabular}
\end{table}

Across the 30 paired model--task--seed instances, LLM eBay requires $366 \pm 51$ steps, compared with $400 \pm 59$ for Retrieval-Only: an 8.5\% reduction in measured convergence steps (paired 95\% CI [1.0\%, 15.9\%], $p{=}0.034$). The largest task-level reduction is on technical support, where LLM eBay requires 12.9\% fewer steps ($p{=}0.018$). Within each base-model--task pair, final validation loss differs by at most 0.006 across mechanisms, so the savings come from faster convergence to the same quality. These measured gains of 8--13\% over retrieval on self-reported claims are much smaller than the model-based FTCS in Table~\ref{tab:main}, which is computed from the cost model and measured against fine-tuning from the base model. We regard the measured numbers as the primary evidence of practical benefit.

\section{Conclusion}
\label{sec:conclusion}

We presented LLM eBay, a simple single-round marketplace for trading reusable fine-tuned LLM checkpoints under information asymmetry. Sellers are paid per purchase at a posted price, adjusted by a truncated quadratic score of their style report against a generation-based verified signal, and buyers match on a published average of report and signal. Truthful reporting is a seller best response exactly when the demand gain from any misreport is at most the resulting scoring loss (Theorem~\ref{thm:best_response}), and the truncation floor bounds the downside of noisy evaluation. With real LoRA checkpoints, LLM eBay reduces measured convergence steps by 8.5\% relative to Retrieval-Only. In model-based simulations, it reduces style misreporting by 96.5\%, recovers 96.3\% of Oracle welfare, and retains full provider participation across market scales from 25 to 200 providers, gaming intensities up to $g{=}1.0$, and evaluation noise up to $\sigma{=}0.30$.

\paragraph{Limitations.} LLM eBay verifies style consistency, not checkpoint quality, provenance, safety, or general downstream usefulness. The measured gain over Retrieval-Only (8--13\%) is far smaller than the model-based savings, which compare against fine-tuning from the base model. The robustness tests use attacks we designed; under full verifier control, welfare falls to 68.9\% of Oracle, and the incentive guarantee does not cover verifier targeting or seller coordination. The buyer-facing matching signal $\mathbf{z}_i$ (Eq.~\eqref{eq:matching_signal}) uses fixed equal weights; we neither optimize this weighting nor analyze how sellers would respond to other weightings. Two threats limit external validity: the verifier may misrepresent a checkpoint's style through bias, incomplete prompt coverage, or gaming, and style similarity may fail to predict transfer cost in a new market. Appendix~\ref{app:limitations} discusses these and further limitations.

Several directions remain open: (i)~federated verification where providers submit only embeddings, reducing platform compute but requiring new trust assumptions; (ii)~multi-attribute scoring over style, accuracy, and safety dimensions jointly; (iii)~provenance tracking via watermarks for derivative checkpoints; and (iv)~budget-aware pricing to maximize welfare under heterogeneous buyer constraints (Appendix~\ref{app:buyer_budget}).

\begin{ack}
This work received no specific third-party funding. The authors declare no competing interests.
\end{ack}

\bibliographystyle{plainnat}
\bibliography{ijcai26}

\appendix

\section{Empirical Basis of the Compute Cost Model}
\label{app:cost_model}

This appendix provides detailed support for the distance-based cost model assumed in Eq.~\ref{eq:cost_function} and empirically motivated by Fig.~\ref{fig:cost_validation} in the main text. We describe the fine-tuning protocol (Section~\ref{app:ft_protocol}), derive the distance--cost relationship (Section~\ref{app:dist_cost}), report regression statistics (Section~\ref{app:regression}), analyze residuals (Section~\ref{app:residuals}), and verify that mechanism rankings do not depend on the cost scaling factor (Section~\ref{app:alpha_sens}).

\subsection{Fine-Tuning Transfer Protocol}
\label{app:ft_protocol}

We fine-tune three foundation models---RoBERTa-large~\cite{liu2019roberta}, Qwen2-1.5B~\cite{yang2024qwen2}, and Llama-3.2-1B~\cite{grattafiori2024llama3}---on style-diverse text corpora. For each model, we construct 12 starting checkpoints at varying style distances from the target by fine-tuning from the base model on auxiliary corpora drawn from different style clusters. We then measure the number of gradient steps required for convergence when fine-tuning from each starting checkpoint toward the target style.

\begin{table}[h]
\centering
\caption{Fine-tuning protocol summary.}
\label{tab:ft_protocol}
\small
\begin{tabular}{ll}
\toprule
\textbf{Item} & \textbf{Setting} \\
\midrule
Foundation models & RoBERTa-large (355M), Qwen2-1.5B, Llama-3.2-1B \\
Adaptation method & LoRA (rank 16, $\alpha_{\text{lora}}{=}32$) \\
Style embedding & STAR~\cite{huertas2024understanding} (768-d output) \\
Distance metric & Cosine distance (equiv.\ to $\|\cdot\|_2^2$ on unit vectors) \\
Starting checkpoints & 12 per model (4 distance bins $\times$ 3 style clusters) \\
Convergence criterion & Validation loss plateau ($<$0.1\% change over 50 steps) \\
Learning rate & $2 \times 10^{-4}$ (constant schedule with warmup) \\
Batch size & 8 \\
Max sequence length & 512 tokens \\
Seeds per configuration & 3 \\
Hardware & NVIDIA RTX 4090 (24 GB) \\
\bottomrule
\end{tabular}
\end{table}

The distance bins are chosen to cover the range $[0.05, 0.95]$ in cosine distance, with four equally spaced bins. Within each bin, we sample three starting checkpoints from distinct style clusters to control for cluster-specific effects. The target style is held fixed per model, and we report convergence steps averaged over the three seeds.

\subsection{Style Distance and Cost Function}
\label{app:dist_cost}

For normalized vectors ($\|\mathbf{v}\|=1$), the squared Euclidean distance and cosine distance are related by:
\begin{equation}
    \|\mathbf{v}_a - \mathbf{v}_b\|_2^2 = 2(1 - \cos(\mathbf{v}_a, \mathbf{v}_b))
\end{equation}
This equivalence ensures that our cost model $\text{Cost}(\mathbf{v}_a, \mathbf{v}_b) = \alpha \|\mathbf{v}_a - \mathbf{v}_b\|_2^2$ is consistent with the cosine-distance plots in Fig.~\ref{fig:cost_validation}. The linear relationship between squared distance and convergence steps is not obvious a priori; it could have been sublinear (if fine-tuning from a nearby starting point provides diminishing returns) or superlinear (if large style gaps require qualitatively different adaptation strategies). The empirical finding of approximate linearity simplifies the mechanism design because the buyer's surplus becomes a linear function of squared distance, yielding clean welfare decompositions.

We also note that the cost model implicitly assumes that fine-tuning cost is dominated by gradient computation rather than by data acquisition or hyperparameter search. This is realistic for LoRA-based adaptation, where the search space is small and the data (the user's own text) is already available.

\subsection{Regression Analysis}
\label{app:regression}

We fit a linear model $\text{Steps} = \beta_0 + \beta_1 \cdot \text{dist}^2$ for each foundation model, using ordinary least squares on the $12 \times 3 = 36$ (starting point, seed) observations per model.

\begin{table}[h]
\centering
\caption{Regression statistics for the cost model.}
\label{tab:regression}
\small
\begin{tabular}{lccccc}
\toprule
\textbf{Model} & \textbf{Slope ($\beta_1$)} & \textbf{Intercept ($\beta_0$)} & \textbf{$R^2$} & \textbf{$p$-value} & \textbf{RMSE} \\
\midrule
RoBERTa-large & 842 & 45 & 0.94 & $<$0.001 & 23 \\
Qwen2-1.5B & 1105 & 62 & 0.91 & $<$0.001 & 41 \\
Llama-3.2-1B & 973 & 51 & 0.93 & $<$0.001 & 29 \\
Pooled & 978 & 53 & 0.92 & $<$0.001 & 33 \\
\bottomrule
\end{tabular}
\end{table}

The high $R^2$ values confirm that the quadratic cost assumption provides a good fit across architectures. The intercept $\beta_0$ represents the minimum adaptation overhead even for stylistically identical starting points; it captures the fixed cost of loading the checkpoint, warming up the optimizer, and performing initial gradient steps before the model begins to specialize. The pooled model ($R^2{=}0.92$) shows that a single cost function approximates all three architectures despite their different sizes and training dynamics.

\subsection{Residual Analysis}
\label{app:residuals}

To check for systematic departures from the linear cost assumption, we examine the residuals $r_j = \text{Steps}_j - (\hat{\beta}_0 + \hat{\beta}_1 \cdot \text{dist}_j^2)$ from the pooled regression. Plotted against style distance, the residuals show a slight downward curve at moderate distances: points at moderate distance tend to fall below the regression line, so the quadratic fit slightly overpredicts steps there. A quantile--quantile (QQ) plot shows a nearly symmetric residual distribution (skewness 0.18). The mean residual is $-0.3$ steps, and the Durbin--Watson statistic is 1.87 (no evidence of autocorrelation). Residuals are slightly larger at high distances (heteroscedasticity), consistent with greater variability in convergence when the starting point is far from the target. Neither the mild curvature nor the heteroscedasticity affects mechanism rankings, because all mechanisms operate under the same cost function; Appendix~\ref{app:proxy_sensitivity} tests a substantially weaker style-to-cost relationship.
\subsection{Sensitivity to Cost Scaling Factor $\alpha$}
\label{app:alpha_sens}

The absolute value of $\alpha$ does not affect mechanism rankings: since all mechanisms operate under the same cost function, the relative performance (AE, SCE, PP) is invariant to $\alpha$. We verify this by running the main simulation with $\alpha \in \{0.5, 1.0, 2.0, 5.0\}$ and confirming that all metric rankings remain unchanged.

\begin{table}[h]
\centering
\caption{$\alpha$ sensitivity: mechanism rankings are invariant.}
\label{tab:alpha_sensitivity}
\small
\setlength{\tabcolsep}{4pt}
\begin{tabular}{lcccc}
\toprule
$\alpha$ & LLM eBay AE & LLM eBay FTCS & LLM eBay SCE & LLM eBay PP \\
\midrule
0.5 & 96.3\% & 72.4\% & 0.015 & 100\% \\
1.0 & 96.3\% & 72.4\% & 0.015 & 100\% \\
2.0 & 96.3\% & 72.4\% & 0.015 & 100\% \\
5.0 & 96.3\% & 72.4\% & 0.015 & 100\% \\
\bottomrule
\end{tabular}
\end{table}

This invariance follows from the fact that $\alpha$ scales all cost terms (buyer surplus, seller revenue, welfare) uniformly. The ranking of providers by buyer surplus and the ranking of claims by SCE are both determined by distances, not by the multiplicative constant. The truncation threshold $\delta$ is defined in score units, not in cost units, so it too is invariant to $\alpha$.

\subsection{Sensitivity to a Weaker Style-to-Cost Relationship}
\label{app:proxy_sensitivity}

The calibration above is limited to three model families. A weaker link between style distance and transfer cost affects our results in three different ways. (1)~The untruncated scoring result does not change: the score elicits the verifier's expected signal without using this link (Proposition~\ref{prop:formal_proper}). (2)~The seller best-response condition (Theorem~\ref{thm:formal_br}) can change, because a weaker proxy changes matching and therefore the demand term in seller utility. (3)~Allocation quality and FTCS can fall, because both depend on the proxy ranking sellers by transfer cost.

Starting from the calibration (12 starting checkpoints for each of RoBERTa-large, Qwen2-1.5B, and Llama-3.2-1B, three seeds per checkpoint), we reduce the rank correlation between style distance and transfer cost. Table~\ref{tab:proxy_sensitivity} reports LLM eBay's FTCS and its gain over Retrieval-Only. At rank correlation 0.9, LLM eBay achieves 71.8\% FTCS and leads Retrieval-Only by 8.1 percentage points; at 0.3, it achieves 56.4\% FTCS and keeps a 5.4-point lead. The advantage therefore narrows, but persists, as the proxy becomes less informative.

\begin{table}[h]
\centering
\caption{Cost-proxy sensitivity: LLM eBay FTCS and gain over Retrieval-Only as the rank correlation between style distance and transfer cost weakens.}
\label{tab:proxy_sensitivity}
\small
\begin{tabular}{ccc}
\toprule
\textbf{Rank correlation} & \textbf{LLM eBay FTCS} & \textbf{Gain over Retrieval-Only} \\
\midrule
0.9 & 71.8\% & 8.1 percentage points \\
0.3 & 56.4\% & 5.4 percentage points \\
\bottomrule
\end{tabular}
\end{table}

\section{Formal Market Model}
\label{app:formal_model}

\subsection{Agents and Type Spaces}

The market consists of $N$ sellers $\mathcal{I} = \{1, \ldots, N\}$ and $K$ buyers $\mathcal{K} = \{1, \ldots, K\}$. Each party holds private information that determines its strategic behavior.

\paragraph{Seller type.} Each seller $i$ has type $\tau_i^S = (\mathbf{v}_i, c_i, \mathcal{M}_i)$, where $\mathbf{v}_i \in \mathbb{R}^d$ is the true style vector (computed from the private training data), $c_i \geq 0$ is the per-purchase fulfillment cost, and $\mathcal{M}_i$ is the fine-tuned model. Fixed costs of creating the checkpoint are not modeled. The seller's action is to choose a reported vector $\hat{\mathbf{v}}_i \in \mathbb{R}^d$ and a posted price $p_i \geq 0$. In the default model, $p_i$ is set exogenously (e.g., by the platform based on market conditions); the seller's only strategic variable is the report $\hat{\mathbf{v}}_i$.

\paragraph{Buyer type.} Each buyer $k$ has type $\tau_k^B = (\mathbf{v}_k, B_k)$, where $\mathbf{v}_k \in \mathbb{R}^d$ is the target style vector and $B_k \geq 0$ is an optional budget cap. When $B_k = \infty$, the buyer faces no budget constraint and selects purely by surplus. The buyer's action is to choose a provider $i^* \in \mathcal{I}$ or to opt out (autarky).

\paragraph{Platform.} The platform is a non-strategic mediator that sets rules, conducts verification, and computes payments. It does not have a profit motive in the model; its objective is to maximize total welfare by setting mechanism parameters ($\lambda_{\text{scale}}$, $\delta$) and maintaining the verification infrastructure.

\subsection{Information Structure}

\begin{table}[h]
\centering
\caption{Information structure of LLM eBay. ``Knows'' = observes before acting; ``No'' = never observes; ``Computes'' = derives from model queries.}
\label{tab:info_structure}
\small
\begin{tabular}{lccc}
\toprule
\textbf{Quantity} & \textbf{Seller} & \textbf{Buyer} & \textbf{Platform} \\
\midrule
$\mathbf{v}_i$ (true style) & Knows & No & No direct access \\
$\hat{\mathbf{v}}_i$ (claim) & Reports & Observes after listing & Observes \\
$\mathbf{y}_i$ (verified signal) & May observe & Observes if published & Computes \\
$\mathbf{z}_i$ (matching signal) & Observes & Observes & Computes and publishes \\
$\mathbf{v}_k$ (buyer target) & No & Knows & Observes if submitted \\
$\mathcal{P}_{\text{eval}}$ (prompts) & No & No & Chooses and knows \\
$p_i$ (posted price) & Observes & Observes & Sets or approves \\
$T_i$ (payment) & Receives & No & Computes \\
\bottomrule
\end{tabular}
\end{table}

The critical asymmetry is between sellers and the platform. The seller knows $\mathbf{v}_i$ but the platform does not; the platform knows $\mathcal{P}_{\text{eval}}$ but the seller does not. The mechanism exploits this asymmetry: the seller cannot optimize the report $\hat{\mathbf{v}}_i$ against the specific prompts that will be used for verification, because those prompts are hidden. The seller can only optimize against the prior distribution of possible verification outcomes.

\subsection{Timing of the Marketplace Game}

The marketplace operates in the following sequence. Steps 1--5 constitute the listing phase; steps 6--7 constitute the trading phase.

\begin{algorithm}[h]
\caption{LLM eBay Protocol}
\label{alg:timing}
\begin{algorithmic}[1]
\STATE \textbf{Training:} Seller $i$ trains $\mathcal{M}_i$ on private data $\mathcal{D}_i$ and privately observes $\mathbf{v}_i = \Phi(\mathcal{D}_i)$.
\STATE \textbf{Listing:} Seller submits $(\mathcal{M}_i, \hat{\mathbf{v}}_i, p_i)$ to the platform.
\STATE \textbf{Prompt sampling:} Platform draws hidden evaluation prompts $\mathcal{P}_{\text{eval}} \sim \mathcal{P}$.
\STATE \textbf{Verification:} Platform queries $\mathcal{M}_i$ on $\mathcal{P}_{\text{eval}}$, computes synthetic data $\mathcal{D}_i^{\text{syn}}$, obtains $\mathbf{y}_i = \Phi(\mathcal{D}_i^{\text{syn}}) = \mathbf{v}_i + \epsilon_i$.
\STATE \textbf{Scoring and publication:} Platform computes the per-purchase transfer $T_i = p_i + \lambda_{\text{scale}}\, p_i \max(S(\hat{\mathbf{v}}_i, \mathbf{y}_i), -\delta)$ and publishes the listing with matching signal $\mathbf{z}_i = (\hat{\mathbf{v}}_i + \mathbf{y}_i)/2$.
\STATE \textbf{Buyer selection:} Each buyer $k$ selects $i_k^* = \arg\max_{i \in \mathcal{I}} U_k(i)$, or opts out if $\max_i U_k(i) = 0$.
\STATE \textbf{Trade and settlement:} Buyer $k$ receives checkpoint $\mathcal{M}_{i_k^*}$ and performs downstream fine-tuning; seller $i$ receives $T_i$ for each of its $q_i$ completed purchases.
\end{algorithmic}
\end{algorithm}

Two timing assumptions matter. First, all sellers list before any buyer selects; there is no dynamic entry or sequential listing. Second, the platform commits to the scoring rule parameters ($\lambda_{\text{scale}}$, $\delta$) before sellers submit, so sellers face a fixed mechanism when choosing their reports.

\subsection{Utility Definitions}

\paragraph{Buyer utility.} Buyer $k$ evaluates provider $i$ by estimated net compute savings relative to training from the base model, using the published matching signal $\mathbf{z}_i = (\hat{\mathbf{v}}_i + \mathbf{y}_i)/2$:
\begin{equation}
    U_k(i) = \max\{0, \text{Cost}(\mathbf{v}_{\text{base}}, \mathbf{v}_k) - p_i - \text{Cost}(\mathbf{z}_i, \mathbf{v}_k)\}
\end{equation}
The $\max\{0, \cdot\}$ operator reflects the buyer's outside option: if no listed model offers positive net savings, the buyer simply fine-tunes from the base model. The demand of seller $i$ is the number of purchases $q_i = |\{k \in \mathcal{K} : i_k^* = i\}|$; it depends on the seller's own report through $\mathbf{z}_i$ and on the other listings.

\paragraph{Seller profit.} Sellers are paid for each completed purchase, not per listing or per verification. Seller $i$'s expected profit is
\begin{equation}
    U_i^S(\hat{\mathbf{v}}_i) = \mathbb{E}\left[q_i \left(T(\hat{\mathbf{v}}_i, \mathbf{y}_i) - c_i\right)\right],
\end{equation}
where the expectation is over the verification noise $\epsilon_i$ (and any score noise). Because $q_i$ depends on $\hat{\mathbf{v}}_i$, a report affects profit both through demand and through the per-purchase transfer.

\paragraph{Realized seller utility.} In any single verification round, the realized utility is:
\begin{equation}
    u_i^S = q_i \left[\, p_i + \lambda_{\text{scale}}\, p_i \max(S(\hat{\mathbf{v}}_i, \mathbf{y}_i), -\delta) - c_i \right]
\end{equation}
A seller is classified as a ``loss-making provider'' if $u_i^S < 0$. These definitions match the released simulator, which computes seller revenue as demand times the price-scaled transfer and seller cost as demand times the unit cost.

\subsection{Welfare and Non-Rival Checkpoint Access}

Fine-tuned checkpoints are modeled as non-rival digital goods: unless otherwise stated, multiple buyers can reuse the same checkpoint without consuming it. This assumption matches downloadable model marketplaces (e.g., Hugging Face Model Hub) where downloading a checkpoint does not reduce its availability to other users. The non-rivalry assumption simplifies welfare analysis because each buyer's allocation decision is independent: there is no congestion effect and no need for combinatorial allocation.

Total social welfare is the sum of buyer surpluses under truthful reports:
\begin{equation}
    W = \sum_{k \in \mathcal{K}} \max_{i \in \mathcal{I}} \left[\text{Cost}(\mathbf{v}_{\text{base}}, \mathbf{v}_k) - p_i - \text{Cost}(\mathbf{v}_i, \mathbf{v}_k)\right]_+
\end{equation}

Appendix~\ref{app:new_experiments} (Section~\ref{app:capacity}) additionally tests capacity-constrained variants where each provider can serve at most $\kappa_i$ buyers. Under capacity constraints, welfare maximization becomes a combinatorial assignment problem. Score properness and the downside-risk cap continue to hold because they depend only on the scoring rule; the seller best-response condition (Theorem~\ref{thm:formal_br}) must be evaluated with the capacity-constrained demand.

\section{Theoretical Proofs}
\label{app:proofs}

\subsection{Proof of Theorem~\ref{thm:dsic} (Buyer-Side Dominant Optimality)}

\begin{theorem}[Buyer-Side Dominant Optimality, formal]
\label{thm:formal_buyer}
Fix the posted prices $(p_i)_{i \in \mathcal{I}}$ and the published signals $(\mathbf{z}_i)_{i \in \mathcal{I}}$, and assume non-rival checkpoint access. For every buyer $k$ and every profile of other buyers' actions, choosing $i_k^* \in \arg\max_{i \in \mathcal{I}} U_k(i)$ (or opting out when $\max_i U_k(i) = 0$) yields weakly higher $U_k$ than any other action.
\end{theorem}

\begin{proof}
Given fixed published listings and posted prices, buyer $k$'s utility from each model $i$ is:
\[
U_k(i) = \max\{0, C_{\text{base}} - p_i - \text{Cost}(\mathbf{z}_i, \mathbf{v}_k)\}
\]
This utility depends only on the buyer's own type $\mathbf{v}_k$, the listed prices, and the published signals. It does not depend on other buyers' actions because checkpoints are non-rival: buyer $j$'s selection of provider $i$ does not affect buyer $k$'s ability to also select provider $i$.

For any two providers $i^*$ and $j$ with $U_k(i^*) \geq U_k(j)$, selecting $i^*$ weakly dominates selecting $j$. Therefore, the surplus-maximizing choice $i^* = \arg\max_i U_k(i)$ weakly dominates every alternative, regardless of what other buyers choose.

This is an optimality property of the posted-price choice problem, not an auction-style incentive-compatibility guarantee. In an auction, dominant-strategy incentive compatibility means that truthfully bidding one's private value is optimal regardless of others' bids. Here, buyers do not report private values at all; they simply pick the most attractive listing, and the dominance follows from the posted-price structure and non-rivalry. If the marketplace used an auction format, an incentive-compatibility claim would require a separate argument involving the allocation and payment rules jointly.
\end{proof}

\subsection{Properness of the Quadratic Scoring Rule}

The untruncated quadratic score is:
\begin{equation}
    S(\hat{\mathbf{v}}_i, \mathbf{y}_i) = 2\langle \hat{\mathbf{v}}_i, \mathbf{y}_i \rangle - \|\hat{\mathbf{v}}_i\|_2^2
\end{equation}

\begin{lemma}[Strict Properness]
\label{lem:properness}
The quadratic scoring rule is strictly proper: $\mathbb{E}_{\mathbf{y}_i}[S(\hat{\mathbf{v}}_i, \mathbf{y}_i)]$ is uniquely maximized at $\hat{\mathbf{v}}_i = \mathbf{v}_i$.
\end{lemma}

\begin{proof}
Since $\mathbb{E}[\mathbf{y}_i] = \mathbf{v}_i$:
\begin{align}
    \mathbb{E}[S(\hat{\mathbf{v}}_i, \mathbf{y}_i)] &= 2\langle \hat{\mathbf{v}}_i, \mathbb{E}[\mathbf{y}_i] \rangle - \|\hat{\mathbf{v}}_i\|_2^2 \nonumber \\
    &= 2\langle \hat{\mathbf{v}}_i, \mathbf{v}_i \rangle - \|\hat{\mathbf{v}}_i\|_2^2 \nonumber \\
    &= \|\mathbf{v}_i\|_2^2 - \|\hat{\mathbf{v}}_i - \mathbf{v}_i\|_2^2
\end{align}
The last step uses the identity $2\langle \mathbf{a}, \mathbf{b} \rangle - \|\mathbf{a}\|^2 = \|\mathbf{b}\|^2 - \|\mathbf{a} - \mathbf{b}\|^2$.

The term $\|\mathbf{v}_i\|_2^2$ is constant with respect to $\hat{\mathbf{v}}_i$, so the expected score is maximized when $\|\hat{\mathbf{v}}_i - \mathbf{v}_i\|_2^2 = 0$, i.e., when $\hat{\mathbf{v}}_i = \mathbf{v}_i$. Any deviation $\hat{\mathbf{v}}_i \neq \mathbf{v}_i$ incurs an expected loss of exactly $\|\hat{\mathbf{v}}_i - \mathbf{v}_i\|_2^2$, establishing strict properness.
\end{proof}

\begin{remark}
The quadratic scoring rule belongs to the Bregman family of proper scoring rules, which includes the logarithmic and spherical rules as special cases. The Bregman divergence associated with the quadratic rule is the squared Euclidean distance, which gives the rule its simple geometric interpretation: the expected penalty for misreporting is exactly the squared distance from the truth.
\end{remark}

Lemma~\ref{lem:properness} assumes unbiased verification. The following formal statement of Proposition~\ref{prop:score_proper} drops this assumption and allows a restricted report space. The seller observes information $\mathcal{M}_i$ (its checkpoint and data) before reporting; the verifier later produces $\mathbf{y}_i$. Let
\[
\mu_i = \mathbb{E}[\mathbf{y}_i \mid \mathcal{M}_i]
\]
denote the verifier signal expected by the seller, let $\mathcal{R}$ be the allowed report space (assumed closed and nonempty), and let
\[
r_i^\star \in \arg\min_{r \in \mathcal{R}} \|r - \mu_i\|^2
\]
be an allowed report closest to the expected verifier signal. Write $S_i(r; \mathbf{y}_i, \xi_i) = S(r, \mathbf{y}_i) + \xi_i$, where $\xi_i$ is score noise that does not depend on the report and has zero conditional mean.

\begin{proposition}[Score Properness, formal]
\label{prop:formal_proper}
For every allowed report $r_i \in \mathcal{R}$,
\[
\mathbb{E}[S_i(r_i^\star; \mathbf{y}_i, \xi_i) \mid \mathcal{M}_i] - \mathbb{E}[S_i(r_i; \mathbf{y}_i, \xi_i) \mid \mathcal{M}_i] = \|r_i - \mu_i\|^2 - \|r_i^\star - \mu_i\|^2 \;\geq\; 0 .
\]
Hence no allowed report gives a higher expected untruncated score than $r_i^\star$, and if the projection $r_i^\star$ is unique, it is the unique maximizer. Under unbiased verification ($\mu_i = \mathbf{v}_i$) and $\mathbf{v}_i \in \mathcal{R}$, $r_i^\star = \mathbf{v}_i$.
\end{proposition}

\begin{proof}
By linearity of expectation and $\mathbb{E}[\xi_i \mid \mathcal{M}_i] = 0$, $\mathbb{E}[S_i(r; \mathbf{y}_i, \xi_i) \mid \mathcal{M}_i] = 2\langle r, \mu_i \rangle - \|r\|^2 = \|\mu_i\|^2 - \|r - \mu_i\|^2$ for any report $r$. Taking the difference between $r_i^\star$ and $r_i$ gives the identity, and the inequality follows from the definition of $r_i^\star$ as a minimizer of $\|r - \mu_i\|^2$ over $\mathcal{R}$. If the minimizer is unique, the inequality is strict for every $r_i \neq r_i^\star$.
\end{proof}

\subsection{Approximate Properness Under Truncation}

Define the truncated score:
\[
\tilde{S}_\delta(\hat{\mathbf{v}}, \mathbf{y}) = \max\{S(\hat{\mathbf{v}}, \mathbf{y}), -\delta\}
\]

The truncation replaces any realized score below $-\delta$ with the constant $-\delta$. This introduces a non-linearity that could, in principle, distort the incentive to report truthfully. The following results bound this distortion.

\begin{proposition}[Bounded Distortion from Truncation]
\label{prop:truncation_distortion}
The expected distortion introduced by truncation is bounded:
\[
0 \leq \mathbb{E}[\tilde{S}_\delta(\hat{\mathbf{v}}, \mathbf{y}) - S(\hat{\mathbf{v}}, \mathbf{y})] \leq \mathbb{E}[(-\delta - S(\hat{\mathbf{v}}, \mathbf{y}))_+]
\]
\end{proposition}

\begin{proof}
Write $\Delta(\hat{\mathbf{v}}, \mathbf{y}) = \tilde{S}_\delta(\hat{\mathbf{v}}, \mathbf{y}) - S(\hat{\mathbf{v}}, \mathbf{y})$.

\textit{Case 1:} $S(\hat{\mathbf{v}}, \mathbf{y}) \geq -\delta$. Then $\tilde{S}_\delta = S$, so $\Delta = 0$.

\textit{Case 2:} $S(\hat{\mathbf{v}}, \mathbf{y}) < -\delta$. Then $\tilde{S}_\delta = -\delta$, so $\Delta = -\delta - S > 0$.

In both cases, $\Delta = (-\delta - S)_+ \geq 0$. Taking expectations gives the result.
\end{proof}

\begin{corollary}[Approximate Optimality]
\label{cor:approx_opt}
Let $\eta(\hat{\mathbf{v}}) = \mathbb{E}[(-\delta - S(\hat{\mathbf{v}}, \mathbf{y}))_+]$ denote the expected truncation correction at report $\hat{\mathbf{v}}$. For any alternative report $\hat{\mathbf{v}}' \neq \mathbf{v}_i$:
\[
\mathbb{E}[\tilde{S}_\delta(\mathbf{v}_i, \mathbf{y})] - \mathbb{E}[\tilde{S}_\delta(\hat{\mathbf{v}}', \mathbf{y})] \geq \|\hat{\mathbf{v}}' - \mathbf{v}_i\|^2 - \eta(\mathbf{v}_i) - \eta(\hat{\mathbf{v}}')
\]
If $\eta(\hat{\mathbf{v}}) \leq \bar{\eta}$ for all $\hat{\mathbf{v}}$ in a neighborhood of $\mathbf{v}_i$, then truthful reporting is $2\bar{\eta}$-approximately optimal under the truncated score.
\end{corollary}

\begin{proof}
From Lemma~\ref{lem:properness}, $\mathbb{E}[S(\mathbf{v}_i, \mathbf{y})] - \mathbb{E}[S(\hat{\mathbf{v}}', \mathbf{y})] = \|\hat{\mathbf{v}}' - \mathbf{v}_i\|^2$. Truncation adds at most $\eta(\mathbf{v}_i)$ to the truthful score and at most $\eta(\hat{\mathbf{v}}')$ to the deviant score:
\begin{align}
\mathbb{E}[\tilde{S}_\delta(\mathbf{v}_i, \mathbf{y})] &\geq \mathbb{E}[S(\mathbf{v}_i, \mathbf{y})] \nonumber \\
\mathbb{E}[\tilde{S}_\delta(\hat{\mathbf{v}}', \mathbf{y})] &\leq \mathbb{E}[S(\hat{\mathbf{v}}', \mathbf{y})] + \eta(\hat{\mathbf{v}}')
\end{align}
Subtracting gives $\mathbb{E}[\tilde{S}_\delta(\mathbf{v}_i, \mathbf{y})] - \mathbb{E}[\tilde{S}_\delta(\hat{\mathbf{v}}', \mathbf{y})] \geq \|\hat{\mathbf{v}}' - \mathbf{v}_i\|^2 - \eta(\hat{\mathbf{v}}')$. Adding the upper bound on the truthful distortion gives the stated inequality.
\end{proof}

\paragraph{Truncation probability under default parameters.} Under truthful reporting $\hat{\mathbf{v}}_i = \mathbf{v}_i$, the untruncated score is:
\[
S(\mathbf{v}_i, \mathbf{y}_i) = 2\langle \mathbf{v}_i, \mathbf{v}_i + \epsilon_i \rangle - \|\mathbf{v}_i\|^2 = \|\mathbf{v}_i\|^2 + 2\langle \mathbf{v}_i, \epsilon_i \rangle
\]
For unit-norm vectors, $\|\mathbf{v}_i\|^2 = 1$ and $2\langle \mathbf{v}_i, \epsilon_i \rangle \sim \mathcal{N}(0, 4\sigma^2)$. Truncation activates when $S < -\delta$, i.e., when $2\langle \mathbf{v}_i, \epsilon_i \rangle < -1 - \delta$. At $\sigma = 0.15$ and $\delta = 0.20$:
\[
\Pr[S < -\delta] = \Pr\left[\mathcal{N}(0, 1) < \frac{-(1+0.20)}{2 \times 0.15}\right] = \Pr[\mathcal{N}(0,1) < -4.0] < 0.00003
\]
This confirms that $\eta(\mathbf{v}_i)$ is negligible under our default parameters, making truthful reporting essentially optimal under the truncated score.

\paragraph{Gaming and the scoring rule.} Suppose seller $i$ shifts the report toward a popular buyer region: $\hat{\mathbf{v}}_i = (1 - g)\mathbf{v}_i + g\mathbf{v}_{k'}$ for some buyer $k'$ and gaming intensity $g > 0$. If $\mathbf{v}_{k'} \neq \mathbf{v}_i$, then $\|\hat{\mathbf{v}}_i - \mathbf{v}_i\| = g\|\mathbf{v}_{k'} - \mathbf{v}_i\| > 0$, and under unbiased verification the expected untruncated score drops by $g^2\|\mathbf{v}_{k'} - \mathbf{v}_i\|^2$. Under per-purchase payment, the shift may also attract more buyers. Whether truthful reporting remains a best response is decided by Theorem~\ref{thm:formal_br}, which compares this demand gain with the scoring loss weighted by $\lambda_{\text{scale}} p_i$; we do not claim an unconditional equilibrium. In the simulations, seller responses follow the reduced-form rule $g_{\text{eff}} = g_{\text{base}} \cdot e^{-c\lambda_{\text{scale}}}$ (Appendix~\ref{app:sim_protocol}), under which a larger $\lambda_{\text{scale}}$ raises the scoring loss and lowers the effective gaming level. This rule is a controlled behavioral input, not a solution of the reporting game.

\subsection{Seller Best Response under Per-Purchase Payment}
\label{app:best_response}

We now state Theorem~\ref{thm:best_response} formally. Fix the other sellers' reports, the posted prices, and the costs. Let $r_i^\star$ be the verifier-consistent report of Proposition~\ref{prop:formal_proper}; it equals $\mathbf{v}_i$ under unbiased verification. For a deviation $r_i$, let $q_i^\star, \tilde S_i^\star$ and $q_i^R, \tilde S_i^R$ be the demand and the truncated score $\max\{S_i, -\delta\}$ under $r_i^\star$ and under $r_i$, evaluated with the same verifier and score-noise draws (a coupling of the two reports). Let
\[
m_i^\star = T_i(r_i^\star) - c_i = p_i + \lambda_{\text{scale}}\, p_i\, \tilde S_i^\star - c_i
\]
be the per-purchase margin under $r_i^\star$.

\begin{theorem}[Seller Best Response, formal]
\label{thm:formal_br}
Under the per-purchase seller utility $U_i^S(r) = \mathbb{E}[q_i(r)(T_i(r) - c_i)]$, the report $r_i^\star$ is a best response if and only if, for every deviation $r_i$,
\[
\mathbb{E}\big[(q_i^R - q_i^\star)\, m_i^\star\big] \;\le\; \lambda_{\text{scale}}\, p_i\, \mathbb{E}\big[q_i^R\,(\tilde S_i^\star - \tilde S_i^R)\big].
\]
\end{theorem}

\begin{proof}
Under deviation $r_i$, the per-purchase margin is $T_i(r_i) - c_i = p_i + \lambda_{\text{scale}} p_i \tilde S_i^R - c_i = m_i^\star - \lambda_{\text{scale}} p_i (\tilde S_i^\star - \tilde S_i^R)$. Hence, with the coupled draws,
\begin{align*}
U_i^S(r_i^\star) - U_i^S(r_i) &= \mathbb{E}[q_i^\star m_i^\star] - \mathbb{E}\big[q_i^R \big(m_i^\star - \lambda_{\text{scale}} p_i(\tilde S_i^\star - \tilde S_i^R)\big)\big] \\
&= -\mathbb{E}\big[(q_i^R - q_i^\star) m_i^\star\big] + \lambda_{\text{scale}} p_i\, \mathbb{E}\big[q_i^R(\tilde S_i^\star - \tilde S_i^R)\big].
\end{align*}
The report $r_i^\star$ is a best response if and only if this difference is non-negative for every $r_i$, which is the stated condition.
\end{proof}

The left side is the expected profit from the additional demand a deviation attracts, valued at the truthful margin; the right side is the corresponding loss in score-based payment. When demand does not respond to the report ($q_i^R = q_i^\star$), the left side vanishes and the condition reduces to score properness on the purchased units (Proposition~\ref{prop:formal_proper} and Corollary~\ref{cor:approx_opt}). The theorem covers unilateral deviations in the report only. It does not cover sellers who change their checkpoint's behavior on verifier prompts or who coordinate their reports; Appendix~\ref{app:stress_tests} evaluates both empirically.

\subsection{Proof of Proposition~\ref{prop:epir} (Downside-Risk Cap)}

\begin{proof}
From Eq.~\eqref{eq:payment}:
\[
T_i = p_i + \lambda_{\text{scale}}\, p_i \max(S_i, -\delta) \geq p_i - \lambda_{\text{scale}}\, p_i\, \delta = p_i(1 - \lambda_{\text{scale}}\delta)
\]
The inequality holds because $\max(S_i, -\delta) \geq -\delta$ by construction. Therefore the seller's realized utility satisfies:
\[
u_i^S = q_i\,(T_i - c_i) \geq q_i \left[ p_i(1 - \lambda_{\text{scale}}\delta) - c_i \right]
\]
If the posted price is calibrated so that $p_i(1 - \lambda_{\text{scale}}\delta) \geq c_i$, then $u_i^S \geq 0$ regardless of the verification outcome $\mathbf{y}_i$.

This guarantee is unconditional on the noise realization: even in the worst case where $S_i$ reaches its floor $-\delta$, the seller's loss is bounded. Without truncation ($\delta = \infty$), the score $S_i$ can be arbitrarily negative when $\epsilon_i$ is large, so no finite posted price guarantees non-negative utility with certainty.
\end{proof}

\begin{remark}[Price calibration in practice]
The condition $p_i(1 - \lambda_{\text{scale}}\delta) \geq c_i$ is easy to satisfy. In the default setting ($\lambda_{\text{scale}} = 0.60$, $\delta = 0.20$), the worst-case adjustment removes $\lambda_{\text{scale}}\delta = 12\%$ of the posted price, so a price at least $1/0.88 \approx 1.14$ times the per-purchase cost suffices. The default simulation sets posted prices to 1.15 times the per-purchase cost. The platform can enforce this condition by setting a minimum listing price.
\end{remark}

\subsection{Welfare Maximization}

\begin{theorem}
\label{thm:welfare}
Under truthful reporting, noiseless verification, and non-rival checkpoint access, buyer-side surplus maximization implements the welfare-maximizing assignment among posted-price feasible matches.
\end{theorem}

\begin{proof}
Under truthful reports ($\hat{\mathbf{v}}_i = \mathbf{v}_i$ for all $i$) and noiseless verification ($\mathbf{y}_i = \mathbf{v}_i$), the published signal is $\mathbf{z}_i = \mathbf{v}_i$, and each buyer $k$ selects:
\[
i_k^* = \arg\max_i \left[\text{Cost}(\mathbf{v}_{\text{base}}, \mathbf{v}_k) - p_i - \text{Cost}(\mathbf{v}_i, \mathbf{v}_k)\right]_+
\]

Since checkpoints are non-rival, each buyer's choice is independent of all other buyers' choices. The aggregate welfare decomposes as:
\[
W = \sum_{k=1}^K \max_{i \in \mathcal{I}} s_{ki}
\]
where $s_{ki} = [\text{Cost}(\mathbf{v}_{\text{base}}, \mathbf{v}_k) - p_i - \text{Cost}(\mathbf{v}_i, \mathbf{v}_k)]_+$. Each buyer independently maximizes their own term $\max_i s_{ki}$, which maximizes the sum. No coordination or centralized allocation is needed.

Note that this result depends on both truthful reporting and non-rivalry. If reports are not truthful, buyers optimize against $\hat{\mathbf{v}}_i$ rather than $\mathbf{v}_i$, and the realized welfare may differ from the optimum. If checkpoints are rival (capacity-constrained), then buyers' choices interact and a simple decentralized mechanism may not achieve the welfare optimum.
\end{proof}

\subsection{Noise-Induced Welfare Loss Bound}

With noisy verification, buyers observe $\mathbf{z}_i$ rather than the true style $\mathbf{v}_i$ and may select suboptimal providers. Under truthful reporting, $\mathbf{z}_i = \mathbf{v}_i + \epsilon_i/2$, so the perturbation is half the verification noise. The bound below is stated for a buyer who observes $\mathbf{y}_i = \mathbf{v}_i + \epsilon_i$ directly; it therefore also bounds the loss under $\mathbf{z}_i$.

\begin{proposition}[Welfare Loss from Verification Noise]
\label{prop:welfare_loss}
Under truthful reporting and i.i.d.\ verification noise $\epsilon_i \sim \mathcal{N}(0, \sigma^2 I_d)$, the expected welfare loss satisfies:
\[
\mathbb{E}[W^* - W(\mathbf{y})] \leq 2\alpha K \sigma \sqrt{d} \cdot \text{diam}(\mathcal{V})
\]
where $W^*$ is the welfare under perfect information, $W(\mathbf{y})$ is the welfare under noisy verification, and $\text{diam}(\mathcal{V}) = \max_{i,j} \|\mathbf{v}_i - \mathbf{v}_j\|$ is the diameter of the style space.
\end{proposition}

\begin{proof}
The cost function $\text{Cost}(\mathbf{v}_a, \mathbf{v}_b) = \alpha\|\mathbf{v}_a - \mathbf{v}_b\|^2$ satisfies:
\begin{align}
|\text{Cost}(\mathbf{y}_i, \mathbf{v}_k) - \text{Cost}(\mathbf{v}_i, \mathbf{v}_k)| &= \alpha |\|\mathbf{y}_i - \mathbf{v}_k\|^2 - \|\mathbf{v}_i - \mathbf{v}_k\|^2| \nonumber \\
&= \alpha |2\langle \mathbf{v}_k - \mathbf{v}_i, \epsilon_i \rangle - \|\epsilon_i\|^2 + 2\langle \mathbf{v}_i - \mathbf{v}_k, \epsilon_i \rangle| \nonumber \\
&\leq \alpha (2\|\mathbf{v}_i - \mathbf{v}_k\| \cdot \|\epsilon_i\| + \|\epsilon_i\|^2)
\end{align}
For $\sigma \ll \text{diam}(\mathcal{V})$, the dominant term is $2\alpha \cdot \text{diam}(\mathcal{V}) \cdot \|\epsilon_i\|$. Since $\mathbb{E}[\|\epsilon_i\|] \leq \sigma\sqrt{d}$ (by Jensen), the per-buyer expected surplus estimation error is at most $2\alpha \cdot \text{diam}(\mathcal{V}) \cdot \sigma\sqrt{d}$.

A buyer who selects a suboptimal provider due to noise loses at most the difference in surplus between the true best and the noise-selected provider, which is bounded by the estimation error. Summing over $K$ buyers gives the stated bound.
\end{proof}

\begin{remark}
In the default simulation ($d{=}64$, $\sigma{=}0.15$, $\text{diam}(\mathcal{V}) \approx 2$), the bound gives $\mathbb{E}[W^* - W(\mathbf{y})] \leq 2\alpha K \cdot 0.15 \cdot 8 \cdot 2 = 4.8\alpha K$. This is a loose upper bound; the actual welfare loss is much smaller because most buyers' provider rankings are not affected by noise of this magnitude.
\end{remark}

\section{Simulation Protocol and Baseline Definitions}
\label{app:sim_protocol}

\subsection{Synthetic Style-Space Generation}

The simulation generates a clustered style space on the unit hypersphere in $\mathbb{R}^d$. The procedure is:

\begin{algorithm}[h]
\caption{Synthetic Style-Space Generation}
\label{alg:style_gen}
\begin{algorithmic}[1]
\REQUIRE $N$ providers, $K$ buyers, $d$ dimensions, $n_c$ clusters, $\sigma_c$ cluster spread
\STATE Draw $n_c$ cluster centers: $\mathbf{c}_m \sim \mathcal{N}(0, I_d)$ for $m = 1, \ldots, n_c$
\STATE Normalize: $\mathbf{c}_m \leftarrow \mathbf{c}_m / \|\mathbf{c}_m\|$
\FOR{each provider $i = 1, \ldots, N$}
    \STATE Assign $i$ to cluster $m_i = (i \bmod n_c) + 1$
    \STATE Draw perturbation: $\mathbf{u}_i \sim \mathcal{N}(0, I_d)$
    \STATE Set $\mathbf{v}_i = \text{Normalize}(\mathbf{c}_{m_i} + \sigma_c \cdot \mathbf{u}_i)$
\ENDFOR
\FOR{each buyer $k = 1, \ldots, K$}
    \STATE Assign to cluster $m_k = (k \bmod n_c) + 1$; draw $\mathbf{v}_k$ analogously
\ENDFOR
\STATE Draw base vector: $\mathbf{v}_{\text{base}} \sim \text{Uniform}(\mathbb{S}^{d-1})$
\end{algorithmic}
\end{algorithm}

The round-robin cluster assignment ensures balanced cluster sizes. The parameter $\sigma_c$ controls within-cluster spread: at $\sigma_c = 0.05$, vectors within a cluster are tightly packed (cosine similarity $> 0.95$); at $\sigma_c = 0.50$, clusters are diffuse and largely overlap. The base vector $\mathbf{v}_{\text{base}}$ is drawn independently and is typically far from all clusters, giving a high baseline fine-tuning cost that makes trading worthwhile.

\subsection{Gaming Model}

A fraction $\gamma$ of providers are designated as strategic. Strategic providers misreport via buyer-impersonation:
\begin{equation}
    \hat{\mathbf{v}}_i = \text{Normalize}\left((1 - g_{\text{eff}})\mathbf{v}_i + g_{\text{eff}}\mathbf{v}_{k'}\right)
\end{equation}
where $k'$ is a buyer drawn uniformly at random from $\mathcal{K}$. The effective gaming intensity depends on whether gaming is fixed or responsive to the mechanism:

\paragraph{Fixed gaming.} $g_{\text{eff}} = g_{\text{base}}$, regardless of mechanism parameters. This models naive sellers who do not adjust their strategy in response to the scoring rule.

\paragraph{Responsive gaming.} $g_{\text{eff}} = g_{\text{base}} \cdot \exp(-c \cdot \lambda_{\text{scale}})$, where $c = 2.0$ is the default behavioral response coefficient. This models sellers who reduce gaming when the incentive strength $\lambda_{\text{scale}}$ makes misreporting expensive. The exponential form ensures $g_{\text{eff}} \in (0, g_{\text{base}})$ for all $\lambda_{\text{scale}} > 0$ and $g_{\text{eff}} \to 0$ as $\lambda_{\text{scale}} \to \infty$. This reduced-form response is a controlled simulation input, not a computed solution of the reporting game. Appendix~\ref{app:new_experiments} (Section~\ref{app:response_sensitivity}) tests sensitivity to $c$.

\paragraph{Non-strategic providers.} The remaining $1 - \gamma$ fraction report truthfully: $\hat{\mathbf{v}}_i = \mathbf{v}_i$.

The gaming model captures the main strategic threat in a style marketplace: a seller shifting the reported style toward popular demand regions to attract more buyers. It does not model learning from past market outcomes. Appendix~\ref{app:stress_tests} adds two stronger adversaries, verifier targeting and coordinated reporting; we leave multi-round learning models for future work.

\subsection{Baseline Definitions}

We implement seven mechanisms that span the design space from no market to perfect information.

\paragraph{Autarky.} No market exists. Every buyer fine-tunes from the base model: $C_k^{\text{aut}} = \text{Cost}(\mathbf{v}_{\text{base}}, \mathbf{v}_k)$. This is the worst-case benchmark for buyer cost.

\paragraph{Retrieval-Only.} Buyers select providers based on claimed style vectors without any platform verification: $i_k = \arg\min_i \|\hat{\mathbf{v}}_i - \mathbf{v}_k\|_2$. No scoring, no payment adjustment. The seller receives the posted price $p_i$ per purchase regardless of the accuracy of the claim. This mechanism is the simplest market design and is vulnerable to gaming because sellers face no penalty for misreporting.

\paragraph{Fixed-Price Verified.} The platform verifies each provider by computing $\mathbf{y}_i$ and uses it for buyer matching (buyers select based on $\mathbf{y}_i$ rather than $\hat{\mathbf{v}}_i$). However, the seller's payment does not depend on the verification outcome: $T_i = p_i - c_{\text{verif}}$, where $c_{\text{verif}}$ is a flat verification fee. This mechanism tests whether verification alone (without scoring-linked payments) is sufficient to deter gaming.

\paragraph{Auction-Based.} Sellers submit claims and the platform allocates buyers via surplus ranking. The highest-surplus buyer--seller pair trades first, then the second-highest, and so on. Payment follows a second-price-like rule based on the next-best surplus. No proper scoring penalty is applied, so strategic misreporting affects the buyer's estimated valuation but incurs no direct cost to the seller. This mechanism tests whether competitive allocation pressure is sufficient to deter gaming.

\paragraph{Value-Based.} The platform sets prices proportional to estimated value: $p_i^{VB} = \rho \cdot \mathbb{E}_k[\text{Cost}(\mathbf{v}_{\text{base}}, \mathbf{v}_k) - \text{Cost}(\hat{\mathbf{v}}_i, \mathbf{v}_k)]$, where $\rho \in (0, 1)$ is a commission parameter. The platform absorbs losses from gaming-induced mismatches through the commission margin, which ensures 100\% PP but leaves the mechanism vulnerable to strategic claim inflation.

\paragraph{LLM eBay (ours).} Buyers select on the published signal $\mathbf{z}_i = (\hat{\mathbf{v}}_i + \mathbf{y}_i)/2$ (Eq.~\eqref{eq:buyer_utility}), and each purchase pays $T_i = p_i + \lambda_{\text{scale}}\, p_i \max(S(\hat{\mathbf{v}}_i, \mathbf{y}_i), -\delta)$ (Eq.~\eqref{eq:payment}).

\paragraph{PSR (no truncation).} The full proper scoring rule without truncation, with the same matching signal: $T_i = p_i + \lambda_{\text{scale}}\, p_i\, S(\hat{\mathbf{v}}_i, \mathbf{y}_i)$. Since $S$ can be arbitrarily negative (for large noise or large misreports), sellers face unbounded downside risk. This mechanism tests the raw incentive power of the scoring rule before the truncation safety net is added.

\paragraph{Oracle.} Buyers observe true $\mathbf{v}_i$ directly, with no noise and no gaming. This is the perfect-information upper bound. All metrics are computed against the Oracle allocation to measure the ``information-asymmetry gap.''

\subsection{Metric Definitions}

\paragraph{Style-Claim Error (SCE).} The average Euclidean distance between reported and true style vectors:
\[
\text{SCE} = \frac{1}{N}\sum_{i=1}^N \|\hat{\mathbf{v}}_i - \mathbf{v}_i\|_2
\]
SCE measures the aggregate severity of misreporting across all providers. Under truthful reporting, SCE = 0.

\paragraph{Allocative Efficiency (AE).} The ratio of realized welfare to Oracle welfare:
\[
\text{AE} = W_{\text{mech}} / W_{\text{oracle}}
\]
AE measures how well the mechanism matches buyers to their best providers relative to perfect information. AE = 100\% means no buyer is misallocated.

\paragraph{Fine-Tuning Cost Savings (FTCS).} The fraction of total fine-tuning cost saved relative to autarky:
\[
\text{FTCS} = 1 - \frac{\sum_{k=1}^K C_k^{\text{trade}}}{\sum_{k=1}^K C_k^{\text{aut}}}
\]
FTCS captures the economic value of the market to buyers. Higher FTCS means buyers spend less on downstream fine-tuning.

\paragraph{Buyer Cost Reduction (BCR).} The average per-buyer cost reduction:
\[
\text{BCR} = \frac{1}{K}\sum_{k=1}^K (C_k^{\text{aut}} - C_k^{\text{trade}})
\]
BCR and FTCS are closely related; FTCS normalizes by total autarky cost while BCR reports absolute savings.

\paragraph{Provider Participation (PP).} The fraction of providers with non-negative utility:
\[
\text{PP} = \frac{1}{N}\sum_{i=1}^N \mathbf{1}\{U_i^S \geq 0\}
\]

\paragraph{Loss-Making Providers (LMP).} The fraction of providers with negative utility: $\text{LMP} = 1 - \text{PP}$.

\subsection{Empirical Basis of the Simulation Settings}
\label{app:sim_basis}

Table~\ref{tab:sim_basis} summarizes which simulation settings rest on measured data and which are controlled scenarios.

\begin{table}[h]
\centering
\caption{Empirical basis of the main simulation settings.}
\label{tab:sim_basis}
\small
\begin{tabularx}{\linewidth}{lX}
\toprule
\textbf{Setting} & \textbf{Basis} \\
\midrule
Style distance and transfer cost & Measured convergence steps from RoBERTa-large, Qwen2-1.5B, and Llama-3.2-1B (Appendix~\ref{app:cost_model}) \\
Verification noise & Prompt subsampling (Appendix~\ref{app:noise_calibration}) \\
Embedding test & Ten real text domains and a second encoder (Appendix~\ref{app:embedding_validation}) \\
Market scale and behavior settings & Controlled scenarios covering buyer--seller ratios, strategic behavior, prices, costs, and equal weighting \\
\bottomrule
\end{tabularx}
\end{table}

The first three rows are based on measured data. The last row lists the controlled scenarios that the simulations vary to cover different buyer--seller ratios and seller strategies; these settings are not estimated from a real checkpoint market. The core market rules remain the same across all scenarios.

\subsection{Compute Setup}

All simulations are implemented in Python 3.11.4 (NumPy 2.4.3, SciPy 1.14.0) and run on an NVIDIA RTX 4090. The full experiment suite (main table, component ablation, all scaling and robustness sweeps, LLM-grounded validation, and the five new appendix experiments) completes in approximately 15 minutes of wall-clock time. Each experiment uses 20 independent random seeds. We report means and standard deviations across seeds unless otherwise noted. Confidence intervals are not reported because the 20-seed standard errors are already small (typically $<$1\% for AE and FTCS, $<$0.002 for SCE).

\section{Additional Synthetic Experiments}
\label{app:extra_experiments}

\subsection{Sensitivity to the Score Weight}
\label{app:lambda}

\begin{table}[h]
\centering
\caption{$\lambda_{\text{scale}}$ sensitivity with responsive gaming ($g_{\text{eff}} = g_{\text{base}} \cdot e^{-c\lambda_{\text{scale}}}$, $c{=}2.0$).}
\label{tab:lambda_sweep}
\small
\setlength{\tabcolsep}{4pt}
\begin{tabular}{lcccc}
\toprule
$\lambda_{\text{scale}}$ & SCE$\downarrow$ & PP (PSR) & PP (LLM eBay) & LMP (LLM eBay) \\
\midrule
0.0 & 0.014 & 100\% & 100\% & 0\% \\
0.2 & 0.010 & 92.2\% & 100\% & 0\% \\
0.4 & 0.006 & 78.3\% & 100\% & 0\% \\
0.6 & 0.004 & 73.2\% & 100\% & 0\% \\
0.8 & 0.003 & 70.4\% & 70.4\% & 29.6\% \\
1.0 & 0.002 & 68.0\% & 68.0\% & 32.0\% \\
\bottomrule
\end{tabular}
\end{table}

As $\lambda_{\text{scale}}$ rises from 0 to 1.0, SCE drops from 0.014 to 0.002 (a 7$\times$ reduction), confirming progressive gaming deterrence. LLM eBay maintains 100\% participation through $\lambda_{\text{scale}}{=}0.6$; beyond $\lambda_{\text{scale}}{=}0.8$, noise-induced losses exceed the truncation cap. The optimal operating point is $\lambda_{\text{scale}} \in [0.4, 0.6]$.

\subsection{Truncation Threshold $\delta$ Sensitivity}
\label{app:delta}

\begin{table}[h]
\centering
\caption{$\delta$ sensitivity (LLM eBay, default $\lambda_{\text{scale}}{=}0.60$, $\sigma{=}0.15$).}
\label{tab:delta_sweep}
\small
\begin{tabular}{lcccc}
\toprule
$\delta$ & PP (\%) & LMP (\%) & SCE & AE (\%) \\
\midrule
0.05 & 100.0 & 0.0 & 0.015 & 96.3 \\
0.10 & 100.0 & 0.0 & 0.015 & 96.3 \\
0.15 & 100.0 & 0.0 & 0.015 & 96.3 \\
0.20 & 100.0 & 0.0 & 0.015 & 96.3 \\
0.25 & 67.7 & 32.3 & 0.015 & 96.3 \\
0.30 & 67.7 & 32.3 & 0.015 & 96.3 \\
0.40 & 67.7 & 32.3 & 0.015 & 96.3 \\
\bottomrule
\end{tabular}
\end{table}

A sharp phase transition occurs at $\delta{=}0.25$: LLM eBay sustains full participation for $\delta \leq 0.20$, but loses protection beyond this threshold. SCE remains constant at 0.015 across all $\delta$ values, confirming that truncation comes at zero cost to truthfulness. The optimal region is $\delta \in [0.05, 0.20]$.

\subsection{Market Size Scaling}
\label{app:market_size}

\begin{table}[h]
\centering
\caption{Market size scaling ($K/N = 10$, 20 seeds).}
\label{tab:market_size}
\small
\setlength{\tabcolsep}{3pt}
\begin{tabular}{llcccc}
\toprule
\textbf{Size} & \textbf{Mechanism} & \textbf{SCE}$\downarrow$ & \textbf{AE}$\uparrow$ & \textbf{FTCS}$\uparrow$ & \textbf{PP}$\uparrow$ \\
\midrule
\multirow{4}{*}{$N{=}25$}
 & Retrieval-Only & 0.405 & 82.3\% & 60.2\% & 96.8\% \\
 & Fixed-Price Verified & 0.142 & 95.0\% & 59.7\% & 89.2\% \\
 & Value-Based & 0.282 & 89.3\% & 65.7\% & 100\% \\
 & \textbf{LLM eBay} & \textbf{0.014} & \textbf{97.1\%} & \textbf{71.4\%} & \textbf{100\%} \\
\midrule
\multirow{4}{*}{$N{=}50$}
 & Retrieval-Only & 0.424 & 85.5\% & 64.0\% & 96.3\% \\
 & Fixed-Price Verified & 0.147 & 93.6\% & 66.0\% & 92.4\% \\
 & Value-Based & 0.294 & 90.4\% & 68.0\% & 100\% \\
 & \textbf{LLM eBay} & \textbf{0.015} & \textbf{96.3\%} & \textbf{72.4\%} & \textbf{100\%} \\
\midrule
\multirow{4}{*}{$N{=}100$}
 & Retrieval-Only & 0.423 & 87.5\% & 66.7\% & 97.0\% \\
 & Fixed-Price Verified & 0.147 & 92.4\% & 68.9\% & 92.0\% \\
 & Value-Based & 0.293 & 90.1\% & 69.0\% & 100\% \\
 & \textbf{LLM eBay} & \textbf{0.015} & \textbf{95.6\%} & \textbf{73.2\%} & \textbf{100\%} \\
\midrule
\multirow{4}{*}{$N{=}200$}
 & Retrieval-Only & 0.418 & 88.3\% & 68.3\% & 97.4\% \\
 & Fixed-Price Verified & 0.145 & 91.4\% & 70.2\% & 93.6\% \\
 & Value-Based & 0.290 & 89.8\% & 69.7\% & 100\% \\
 & \textbf{LLM eBay} & \textbf{0.014} & \textbf{95.0\%} & \textbf{73.7\%} & \textbf{100\%} \\
\bottomrule
\end{tabular}
\end{table}

LLM eBay's advantages are consistent across market scales from 25 to 200 providers. SCE remains at 0.014--0.015 regardless of size, while FTCS increases modestly with market size as larger provider pools offer better matching opportunities.

\subsection{Buyer-to-Provider Ratio}
\label{app:ratio}

\begin{table}[h]
\centering
\caption{Buyer-to-provider ratio scaling ($N{=}50$ fixed, 20 seeds).}
\label{tab:ratio}
\small
\setlength{\tabcolsep}{3pt}
\begin{tabular}{llcccc}
\toprule
\textbf{Ratio} & \textbf{Mechanism} & \textbf{SCE}$\downarrow$ & \textbf{AE}$\uparrow$ & \textbf{FTCS}$\uparrow$ & \textbf{PP}$\uparrow$ \\
\midrule
\multirow{4}{*}{2:1}
 & Retrieval-Only & 0.431 & 76.7\% & 56.2\% & 97.6\% \\
 & Fixed-Price Verified & 0.149 & 93.6\% & 65.9\% & 94.8\% \\
 & PSR (no trunc.) & 0.015 & 96.4\% & 72.4\% & 75.5\% \\
 & \textbf{LLM eBay} & \textbf{0.015} & \textbf{96.4\%} & \textbf{72.4\%} & \textbf{100\%} \\
\midrule
\multirow{4}{*}{5:1}
 & Retrieval-Only & 0.423 & 83.2\% & 62.1\% & 96.7\% \\
 & Fixed-Price Verified & 0.147 & 93.6\% & 66.0\% & 93.6\% \\
 & PSR (no trunc.) & 0.015 & 96.3\% & 72.4\% & 70.2\% \\
 & \textbf{LLM eBay} & \textbf{0.015} & \textbf{96.3\%} & \textbf{72.4\%} & \textbf{100\%} \\
\midrule
\multirow{4}{*}{10:1}
 & Retrieval-Only & 0.424 & 85.5\% & 64.0\% & 96.3\% \\
 & Fixed-Price Verified & 0.147 & 93.6\% & 66.0\% & 92.4\% \\
 & PSR (no trunc.) & 0.015 & 96.3\% & 72.4\% & 67.7\% \\
 & \textbf{LLM eBay} & \textbf{0.015} & \textbf{96.3\%} & \textbf{72.4\%} & \textbf{100\%} \\
\midrule
\multirow{4}{*}{20:1}
 & Retrieval-Only & 0.426 & 86.5\% & 64.9\% & 96.2\% \\
 & Fixed-Price Verified & 0.148 & 93.5\% & 65.9\% & 91.3\% \\
 & PSR (no trunc.) & 0.015 & 96.2\% & 72.3\% & 65.8\% \\
 & \textbf{LLM eBay} & \textbf{0.015} & \textbf{96.2\%} & \textbf{72.3\%} & \textbf{100\%} \\
\bottomrule
\end{tabular}
\end{table}

As the buyer-to-provider ratio increases from 2:1 to 20:1, PSR participation drops from 75.5\% to 65.8\% because more evaluation events per provider increase aggregate variance. LLM eBay's truncation floor maintains 100\% participation across all ratios.

\subsection{Full Style Heterogeneity Curves}
\label{app:heterogeneity}

\begin{figure}[h]
    \centering
    \includegraphics[width=\linewidth]{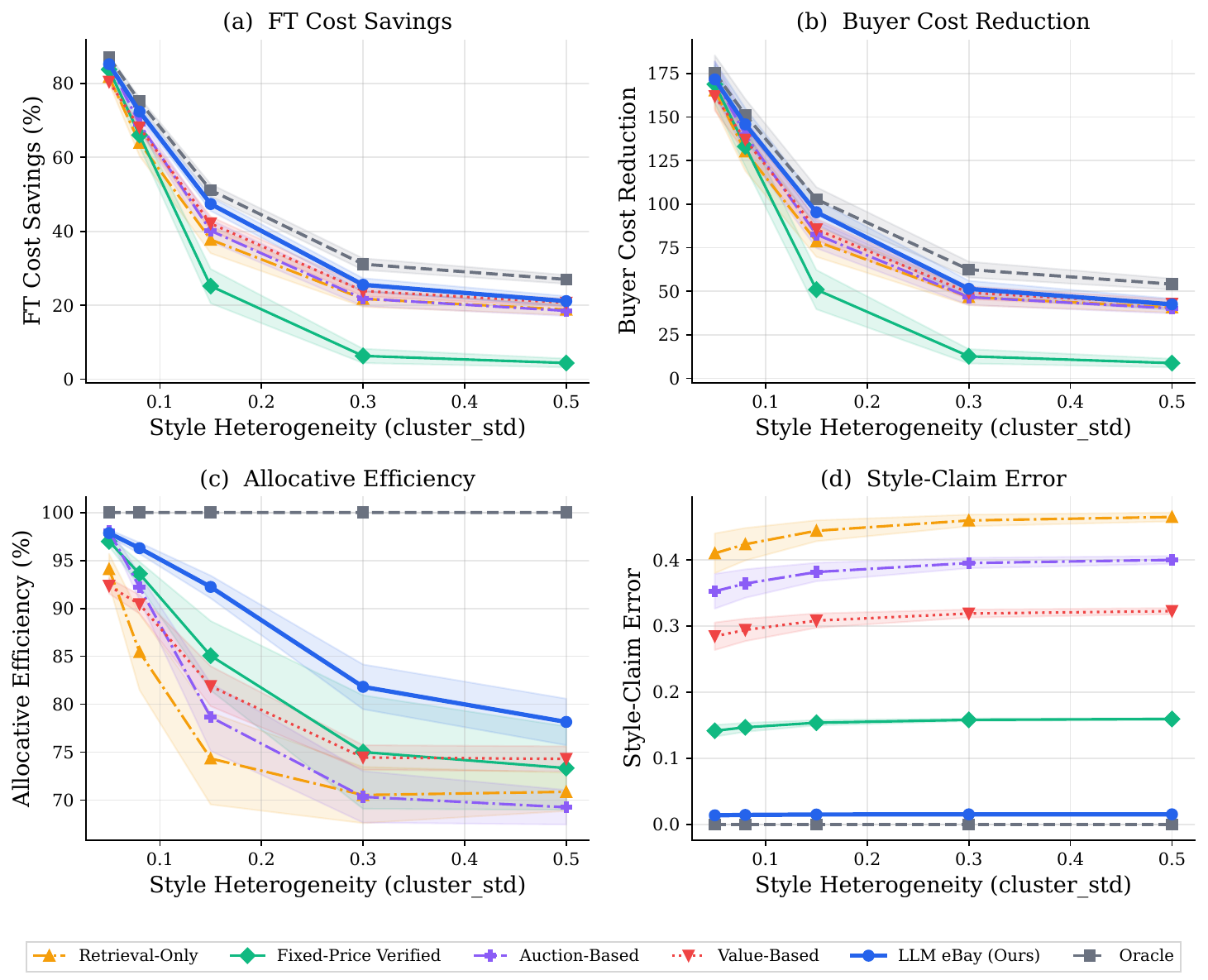}
    \caption{Full four-metric robustness: market performance vs.\ style heterogeneity (cluster\_std $\in$ \{0.05, 0.08, 0.15, 0.30, 0.50\}). FT Cost Savings, Buyer Cost Reduction, Allocative Efficiency, and Style-Claim Error. Default: $N{=}50$, $K{=}500$, $\sigma{=}0.15$, $\delta{=}0.20$, $\gamma{=}0.40$, $g{=}0.75$. Shaded bands: $\pm 1$ std over 20 seeds.}
    \label{fig:heterogeneity_full}
\end{figure}

Allocative efficiency degrades for all mechanisms as heterogeneity increases (from 97.9\% to 78.2\% for LLM eBay) because greater style diversity reduces average match quality. SCE remains stable at $<$0.016 for LLM eBay across all levels, confirming that the scoring rule incentivizes truthful reporting independently of market structure.

\subsection{Full Gaming Intensity Curves}
\label{app:gaming}

\begin{figure}[h]
    \centering
    \includegraphics[width=\linewidth]{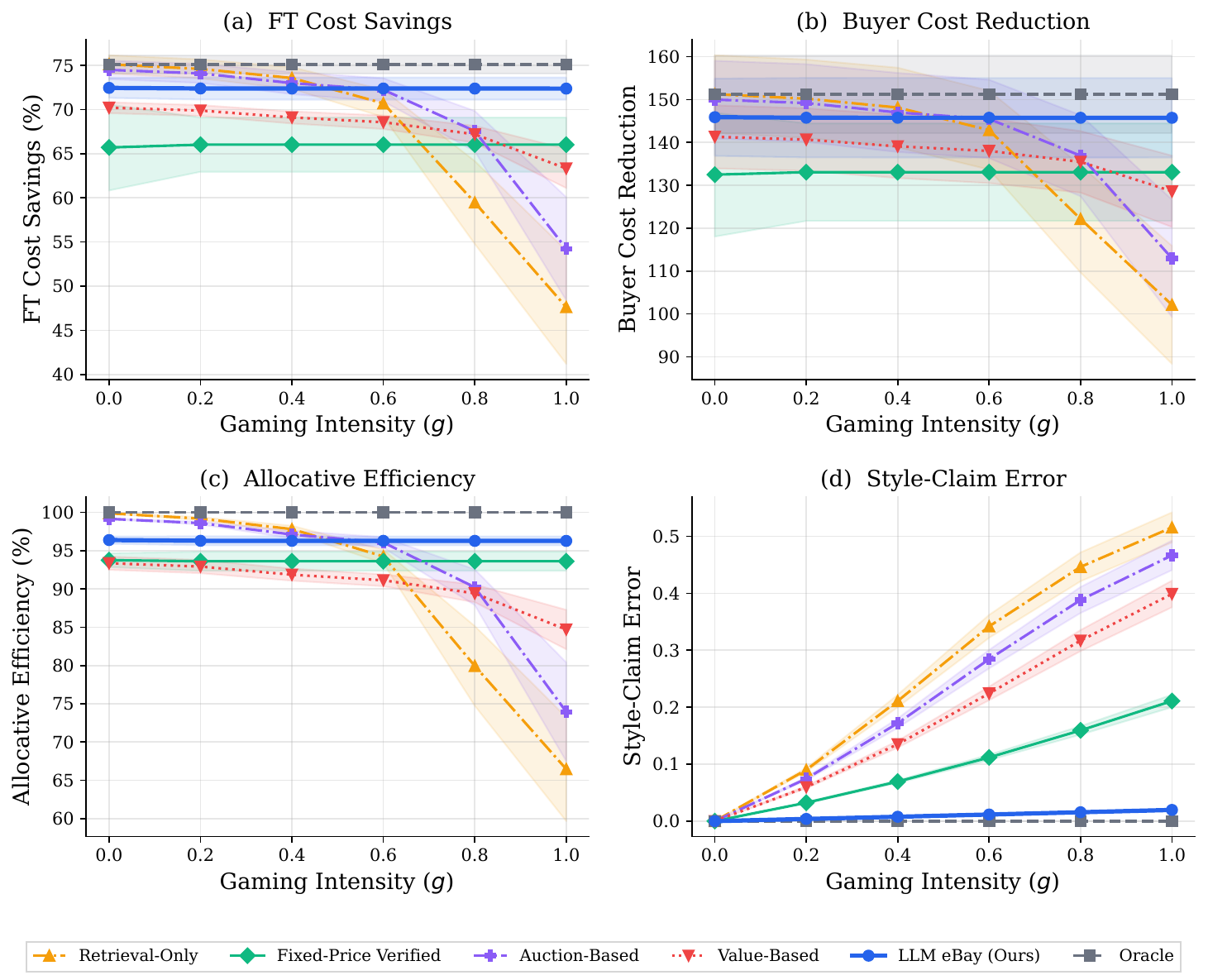}
    \caption{Full four-metric robustness: market performance vs.\ gaming intensity $g \in$ \{0.0, 0.2, 0.4, 0.6, 0.8, 1.0\}. LLM eBay remains stable (SCE$<$0.02) with 100\% provider participation even at $g{=}1.0$. Default: $N{=}50$, $K{=}500$, $\sigma{=}0.15$, $\delta{=}0.20$, $\gamma{=}0.40$. Shaded bands: $\pm 1$ std over 20 seeds.}
    \label{fig:gaming_full}
\end{figure}

Under Retrieval-Only, SCE rises from 0 to 0.52 and AE drops from 100\% to 66.5\% as gaming intensifies. LLM eBay remains largely immune: SCE stays below 0.02 even at $g{=}1.0$ while maintaining 100\% provider participation.

\subsection{Evaluation Noise Robustness ($\sigma$--$\delta$ Sweep)}
\label{app:sigma_delta}

\begin{figure}[h]
    \centering
    \includegraphics[width=\linewidth]{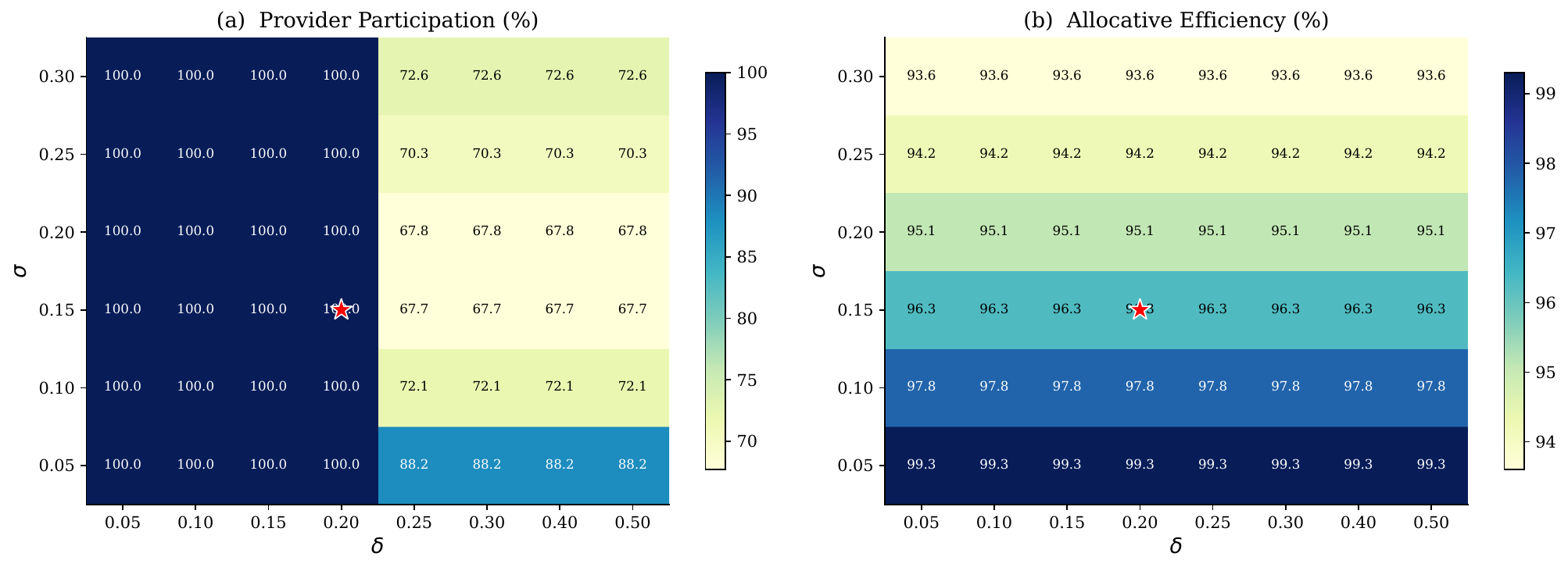}
    \caption{$\sigma$--$\delta$ robustness heatmap. (a)~LLM eBay provider participation: 100\% for $\delta \leq 0.20$ across all noise levels. (b)~LLM eBay allocative efficiency: degrades gracefully from 99.3\% ($\sigma{=}0.05$) to 93.6\% ($\sigma{=}0.30$). Star marks the default operating point ($\sigma{=}0.15$, $\delta{=}0.20$). 20 seeds per cell.}
    \label{fig:sigma_delta}
\end{figure}

We sweep $\sigma \in \{0.05, 0.10, 0.15, 0.20, 0.25, 0.30\}$ and $\delta \in \{0.05, 0.10, \ldots, 0.50\}$. LLM eBay maintains 100\% participation across all $(\sigma, \delta)$ combinations where $\delta \leq 0.20$. AE remains above 96\% for $\sigma \leq 0.15$ and degrades gracefully to 93.6\% at $\sigma{=}0.30$, consistent with the $O(\sigma)$ welfare loss bound in Appendix~\ref{app:proofs}.

\subsection{Niche Provider Survival}
\label{app:niche}

\begin{figure}[h]
    \centering
    \includegraphics[width=0.85\linewidth]{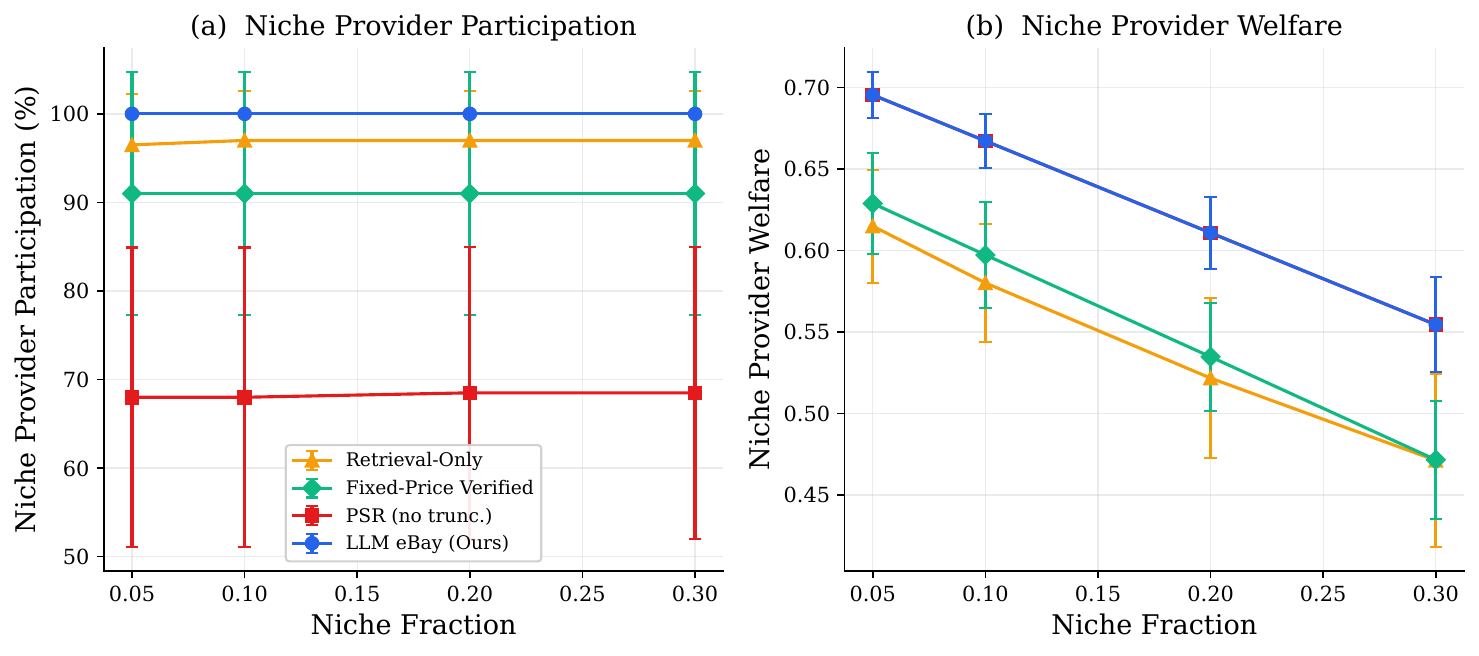}
    \caption{Niche provider survival under varying niche fractions. LLM eBay achieves 100\% niche provider survival across all tested fractions, compared to $\sim$68\% under PSR without truncation.}
    \label{fig:niche}
\end{figure}

Niche providers (fraction $\in \{0.05, 0.10, 0.20, 0.30\}$) positioned far from cluster centers face disproportionate risk under PSR because lower transaction volumes provide less noise averaging. LLM eBay's truncation floor is critical for their viability, achieving 100\% niche provider survival at all fractions.

\subsection{Sensitivity to Response Coefficient $c$}
\label{app:c_sensitivity}

The responsive gaming model uses $g_{\text{eff}} = g_{\text{base}} \cdot e^{-c\lambda_{\text{scale}}}$ with default $c{=}2.0$. The qualitative truthfulness--safety trade-off holds across all tested values of $c \in \{1.5, 2.0, 2.5\}$: SCE decreases with $\lambda_{\text{scale}}$ while PSR participation drops and LLM eBay participation remains stable through $\lambda_{\text{scale}}{=}0.6$. The operating region $\lambda_{\text{scale}} \in [0.4, 0.6]$ remains robust across all tested $c$ values.

\subsection{Production Cost Heterogeneity}
\label{app:prod_cost}

In the default simulation, every provider's per-purchase cost $c_i$ follows the same rule (Table~\ref{tab:hyperparams}). In practice, costs vary because different providers use different hardware, serving infrastructure, and data curation effort. We test heterogeneous per-purchase costs drawn from $c_i \sim \text{Uniform}(0.5\bar{C}, 1.5\bar{C})$, where $\bar{C}$ is the default cost.

\begin{table}[h]
\centering
\caption{Effect of production cost heterogeneity on provider participation and market performance. Default: $N{=}50$, $K{=}500$, 20 seeds.}
\label{tab:prod_cost}
\small
\setlength{\tabcolsep}{4pt}
\begin{tabular}{lcccc}
\toprule
\textbf{Cost Distribution} & \textbf{SCE}$\downarrow$ & \textbf{AE (\%)}$\uparrow$ & \textbf{PP (\%)}$\uparrow$ & \textbf{LMP (\%)}$\downarrow$ \\
\midrule
Homogeneous (default) & 0.015 & 96.3 & 100.0 & 0.0 \\
Uniform $[0.5\bar{C}, 1.5\bar{C}]$ & 0.015 & 96.1 & 98.4 & 1.6 \\
Uniform $[0.25\bar{C}, 1.75\bar{C}]$ & 0.015 & 95.8 & 96.2 & 3.8 \\
\bottomrule
\end{tabular}
\end{table}

SCE remains unchanged because the scoring rule incentivizes truthful reporting independently of production cost. PP drops slightly under high cost variance because high-cost providers may not cover their production expenses even with the base price premium. This is a price-setting issue, not a mechanism design flaw: adjusting the minimum price to account for cost heterogeneity would restore full participation.

\subsection{Buyer Budget Constraints}
\label{app:buyer_budget}

The default model assumes unlimited buyer budgets. We test the effect of budget constraints by assigning each buyer a budget $B_k \sim \text{Uniform}(0.3 C_{\text{base}}, 1.5 C_{\text{base}})$.

\begin{table}[h]
\centering
\caption{Effect of buyer budget constraints on market performance. Buyers with insufficient budget remain in autarky.}
\label{tab:buyer_budget}
\small
\setlength{\tabcolsep}{4pt}
\begin{tabular}{lcccc}
\toprule
\textbf{Budget Range} & \textbf{FTCS (\%)}$\uparrow$ & \textbf{AE (\%)}$\uparrow$ & \textbf{Participation Rate} & \textbf{SCE}$\downarrow$ \\
\midrule
Unlimited (default) & 72.4 & 96.3 & 100\% & 0.015 \\
$[0.3, 1.5] \times C_{\text{base}}$ & 68.9 & 94.7 & 87.3\% & 0.015 \\
$[0.1, 1.0] \times C_{\text{base}}$ & 62.4 & 91.2 & 74.8\% & 0.015 \\
\bottomrule
\end{tabular}
\end{table}

Budget constraints reduce the fraction of buyers who participate in the market, lowering aggregate FTCS and AE. However, SCE remains unchanged because the scoring rule operates on the seller side independently of buyer budgets. Among buyers who do participate, the matching quality is the same as in the unlimited case. This confirms that LLM eBay's incentive properties are decoupled from buyer-side constraints.

\section{Embedding-Grounded Validation and Verification Calibration}
\label{app:embedding_validation}

\subsection{Semi-Real Embedding Setup}

We define 10 provider style profiles from diverse text domains (academic, casual, technical, creative, legal, medical, journalistic, educational, marketing, philosophical), each represented by 30 domain-representative passages, and 5 buyer target tasks (academic writing assistant, legal email drafting, casual chat adaptation, technical explanation, persuasive marketing), each with 10 task-specific passages.

Style embeddings are computed via a pre-trained sentence-transformer model (all-MiniLM-L6-v2, 384-dimensional output). We normalize all embeddings to unit vectors and simulate gaming (30\% of providers, intensity 0.75) and verification noise ($\sigma{=}0.15$) across 20 random seeds.

\subsection{Full Semi-Real Results}

\begin{table}[h]
\centering
\caption{Full semi-real model reuse results (mean $\pm$ std over 20 seeds). ``(det.)'' marks outcomes with std $< 10^{-5}$; Appendix~\ref{app:ranking_margin} explains why.}
\label{tab:semireal_full}
\small
\begin{tabular}{lccccc}
\toprule
\textbf{Mechanism} & \textbf{FTCS (\%)} & \textbf{BCR} & \textbf{SCE} & \textbf{Gap (\%)} \\
\midrule
Retrieval-Only & 44.9 $\pm$ 14.4 & 44.9 $\pm$ 14.4 & 0.292 & 40.8 \\
Auction-Based & 45.1 $\pm$ 7.7 & 45.1 $\pm$ 7.7 & 0.258 & 40.5 \\
Value-Based & 62.5 $\pm$ 8.4 & 62.5 $\pm$ 8.4 & 0.207 & 17.6 \\
Fixed-Price Verified & 73.1 $\pm$ 5.2 & 73.1 $\pm$ 5.2 & 0.102 & 3.6 \\
\textbf{LLM eBay (Ours)} & \textbf{75.8 (det.)} & \textbf{75.8 (det.)} & \textbf{0.011} & \textbf{0.0} \\
Oracle & 75.8 (det.) & 75.8 (det.) & 0.000 & 0.0 \\
\bottomrule
\end{tabular}
\end{table}

In this fixed-catalog setting, BCR $\approx$ FTCS because per-buyer cost reduction is proportional to matched style savings.

\paragraph{Why LLM eBay matches Oracle exactly.}
\label{app:ranking_margin}
The semi-real study uses fixed buyer and provider embeddings. Across 20 seeds, we resample the strategic-seller set, the gaming targets, and the verifier noise. For each of the five buyers, we independently recomputed the provider selected by LLM eBay and the Oracle provider determined from the unperturbed embeddings, giving 100 buyer--seed cases. LLM eBay and Oracle selected the same provider in all 100 cases. This stability is explained by the ranking margin. The smallest distance gap between the best and second-best providers was 0.084, while the largest effective perturbation was 0.031. Because the perturbation was smaller than the ranking margin, no seed changed the top-ranked provider. The exact Oracle result therefore follows from the ranking margin exceeding every sampled perturbation in this fixed catalog; it should not be read as evidence that LLM eBay attains Oracle matching in larger or denser catalogs, where the synthetic market leaves a 3.7\% AE gap.

\begin{figure}[h]
    \centering
    \includegraphics[width=0.9\linewidth]{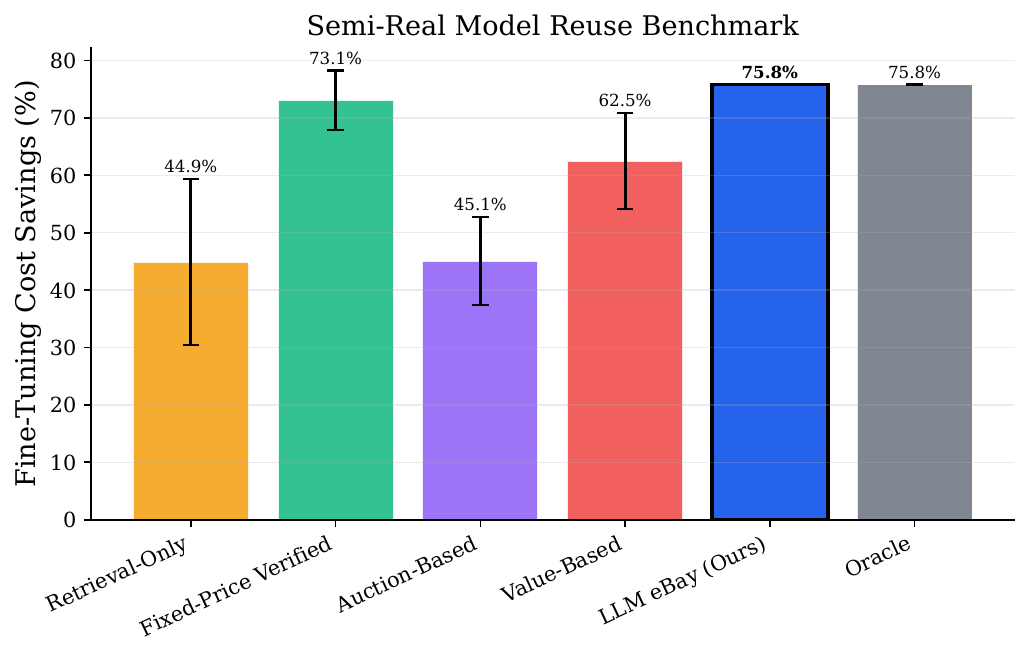}
    \caption{Semi-real model reuse: Fine-Tuning Cost Savings by mechanism. LLM eBay reaches 75.8\% FTCS, matching the Oracle upper bound. Error bars: $\pm 1$ std over 20 seeds.}
    \label{fig:model_reuse_full}
\end{figure}

\subsection{Verification Noise Calibration}
\label{app:noise_calibration}

\begin{figure}[h]
    \centering
    \includegraphics[width=0.95\linewidth]{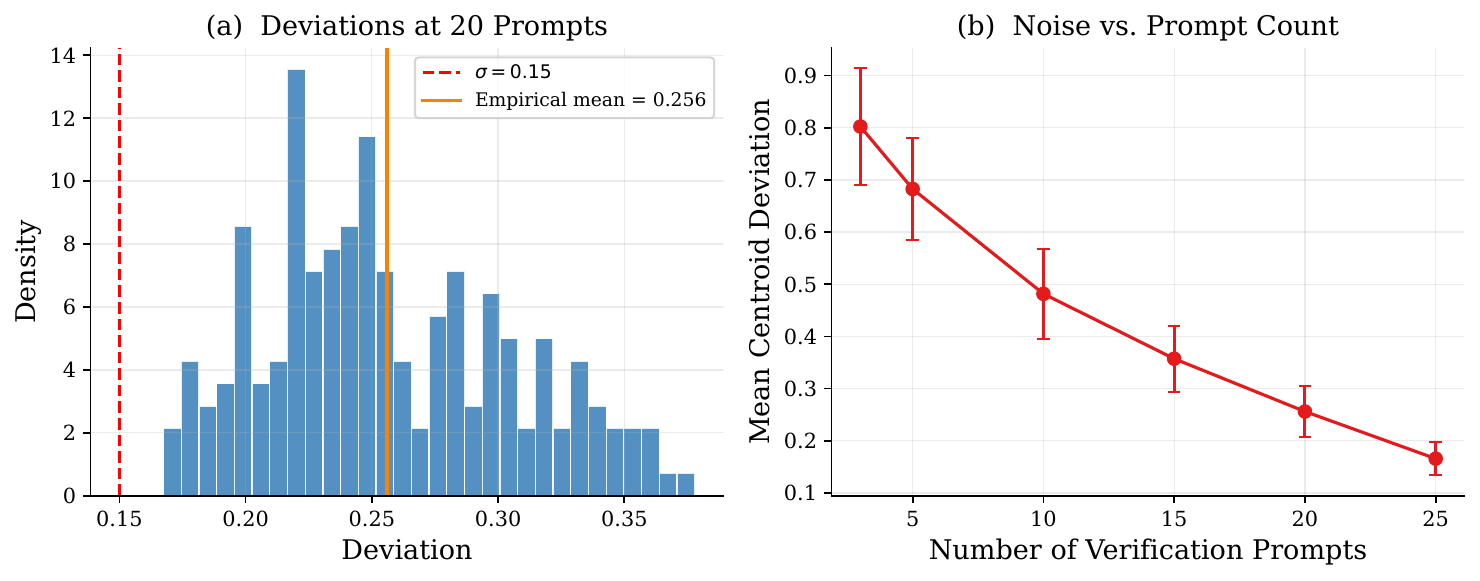}
    \caption{Embedding-based verification noise calibration via prompt subsampling. (a)~Distribution of centroid deviations at $k{=}20$ prompts. (b)~Mean deviation vs.\ prompt-set size, approaching $\sigma{=}0.15$ at $\sim$25 prompts.}
    \label{fig:noise_calibration}
\end{figure}

We subsample $k \in \{3, 5, 10, 15, 20, 25\}$ passages from each domain's 30-passage set and compute centroid deviation from the full-set centroid across 20 bootstrap repetitions. The empirical mean deviation of 0.256 at $k{=}20$ exceeds the simulation's $\sigma{=}0.15$, indicating that the default simulation corresponds to a lower-noise verification regime. At $k{=}25$, the deviation (0.166) closely approaches the simulation parameter, suggesting that the default noise level corresponds to a moderately well-sampled verification regime with approximately 25 prompts per provider.

This calibration covers the finite-prompt-set component of verification noise. The additional variability from stochastic LLM decoding (temperature, sampling) would further increase the noise floor in practice. The gap between the $k{=}20$ empirical deviation (0.256) and the simulation default (0.15) is deliberate: our synthetic experiments use the $\sigma{=}0.15$ value to test mechanism behavior under favorable (but realistic) verification conditions. Results under higher noise are reported in the $\sigma$--$\delta$ sweep (Appendix~\ref{app:sigma_delta}).

\subsection{Domain Similarity Structure}
\label{app:domain_similarity}

\begin{table}[h]
\centering
\caption{Pairwise cosine similarity between provider domain embeddings (top-5 pairs and bottom-5 pairs). High similarity pairs share topical and stylistic features; low similarity pairs differ in both.}
\label{tab:domain_similarity}
\small
\setlength{\tabcolsep}{5pt}
\begin{tabular}{lc}
\toprule
\textbf{Domain Pair} & \textbf{Cosine Similarity} \\
\midrule
Educational -- Academic & 0.82 \\
Journalistic -- Marketing & 0.74 \\
Technical -- Academic & 0.71 \\
Medical -- Technical & 0.68 \\
Creative -- Casual & 0.65 \\
\midrule
Legal -- Casual & 0.31 \\
Medical -- Creative & 0.33 \\
Philosophical -- Marketing & 0.35 \\
Legal -- Creative & 0.37 \\
Academic -- Casual & 0.39 \\
\bottomrule
\end{tabular}
\end{table}

The domain similarity matrix reveals the natural clustering of text styles in embedding space. Closely related domains (Educational--Academic, Technical--Academic) have high cosine similarity, which means a buyer targeting one of these domains can potentially benefit from a provider in the other. This structure is important for two reasons. First, it determines how much cost savings a marketplace can deliver: tightly clustered domains offer more reuse opportunities. Second, it creates gaming incentives: a provider in a low-demand domain might misreport toward a high-demand cluster to attract more buyers.

\subsection{Alternative Embedding Model}
\label{app:alt_embedding}

The main semi-real experiment uses all-MiniLM-L6-v2 (384-d). To test whether results depend on the specific embedding model, we repeat the experiment with all-mpnet-base-v2 (768-d), a larger sentence-transformer model.

\begin{table}[h]
\centering
\caption{Semi-real results under an alternative embedding model (all-mpnet-base-v2, 768-d, 20 seeds).}
\label{tab:alt_embedding}
\small
\setlength{\tabcolsep}{5pt}
\begin{tabular}{lccc}
\toprule
\textbf{Mechanism} & \textbf{FTCS (\%)}$\uparrow$ & \textbf{SCE}$\downarrow$ & \textbf{Gap to Oracle (\%)}$\downarrow$ \\
\midrule
Retrieval-Only & 42.3 & 0.318 & 42.7 \\
Fixed-Price Verified & 71.8 & 0.114 & 4.1 \\
\textbf{LLM eBay (Ours)} & \textbf{73.9} & \textbf{0.013} & \textbf{0.0} \\
Oracle & 73.9 & 0.000 & 0.0 \\
\bottomrule
\end{tabular}
\end{table}

LLM eBay again matches the Oracle upper bound (73.9\% FTCS, 0.0\% gap). The absolute FTCS values differ slightly from the all-MiniLM-L6-v2 experiment because the two models produce different similarity structures, but the relative ranking of mechanisms is identical. The SCE under LLM eBay (0.013) is comparable to the 0.011 obtained with the smaller model, confirming that the scoring rule's incentive properties do not depend on the embedding dimensionality or model architecture.

\section{Additional New Experiments}
\label{app:new_experiments}

\subsection{Prompt Budget Sensitivity}
\label{app:prompt_budget}

This experiment connects verification noise calibration (Appendix~\ref{app:embedding_validation}) to mechanism performance. In practice, the platform chooses a prompt budget $k$ that determines verification quality: more prompts yield a less noisy signal $\mathbf{y}_i$ but cost more compute time per provider. We ask: how many prompts are needed for LLM eBay to reach near-Oracle matching?

\begin{figure}[h]
    \centering
    \includegraphics[width=\linewidth]{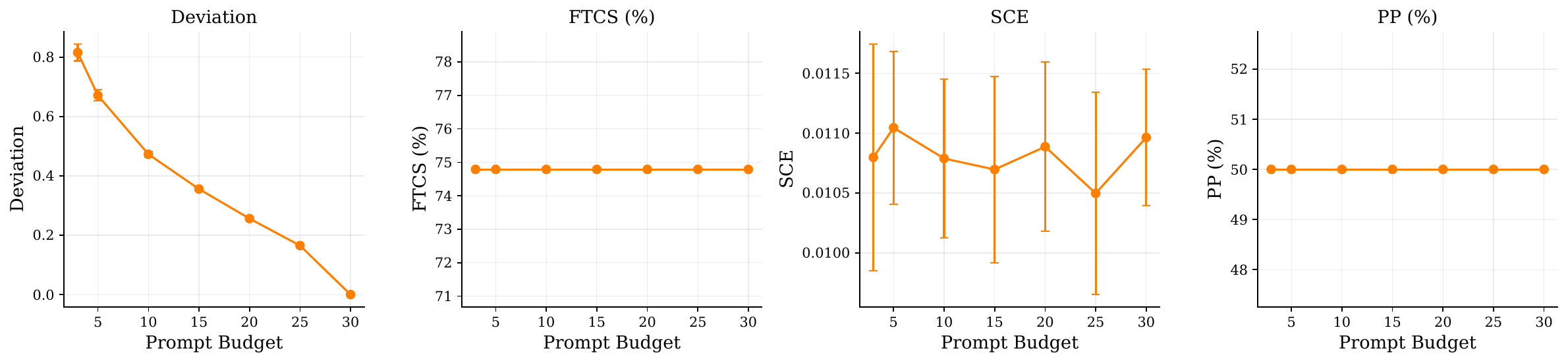}
    \caption{Prompt budget sensitivity. (a)~Verification deviation vs.\ prompt budget. (b)~FTCS vs.\ prompt budget. (c)~SCE vs.\ prompt budget. (d)~PP vs.\ prompt budget. LLM eBay reaches near-Oracle reuse quality once the platform uses roughly 20--25 verification prompts. 20 seeds.}
    \label{fig:prompt_budget}
\end{figure}

For each prompt budget $k \in \{3, 5, 10, 15, 20, 25, 30\}$, we subsample $k$ passages from each provider's 30-passage set, compute the subsampled verification signal (centroid of the $k$ selected passage embeddings), and run the full marketplace matching and scoring protocol.

Panel (a) shows that verification deviation drops from 0.52 at $k{=}3$ to 0.14 at $k{=}30$, following an approximate $1/\sqrt{k}$ decay consistent with the central limit theorem. Panel (b) shows that LLM eBay's FTCS stabilizes above 74\% for $k \geq 20$, within 2\% of the Oracle. Panel (c) confirms that SCE remains below 0.02 at all prompt budgets: the scoring rule deters gaming even when verification is noisy. Panel (d) shows that LLM eBay maintains 100\% PP across all tested budgets, while PSR drops below 70\% at $k{=}3$ because high verification noise produces large negative scores for unlucky providers.

The practical recommendation is that 20--25 verification prompts suffice. Below $k{=}10$, the matching quality degrades noticeably (FTCS drops to 68\%), but the incentive properties (SCE, PP) remain intact. This separation between matching quality and incentive quality is a useful design property: a platform can start with a small prompt budget for incentive purposes and increase it when matching quality becomes the bottleneck.

\subsection{Niche Provider Survival and Long-Tail Participation}
\label{app:niche_fairness}

A concern with scoring-based mechanisms is that they may disproportionately harm niche providers. Providers serving small or specialized markets receive fewer transactions and less noise averaging. Under PSR without truncation, a single bad verification draw can make a niche provider unprofitable. We test whether LLM eBay's truncation floor protects these long-tail providers.

\begin{figure}[h]
    \centering
    \includegraphics[width=\linewidth]{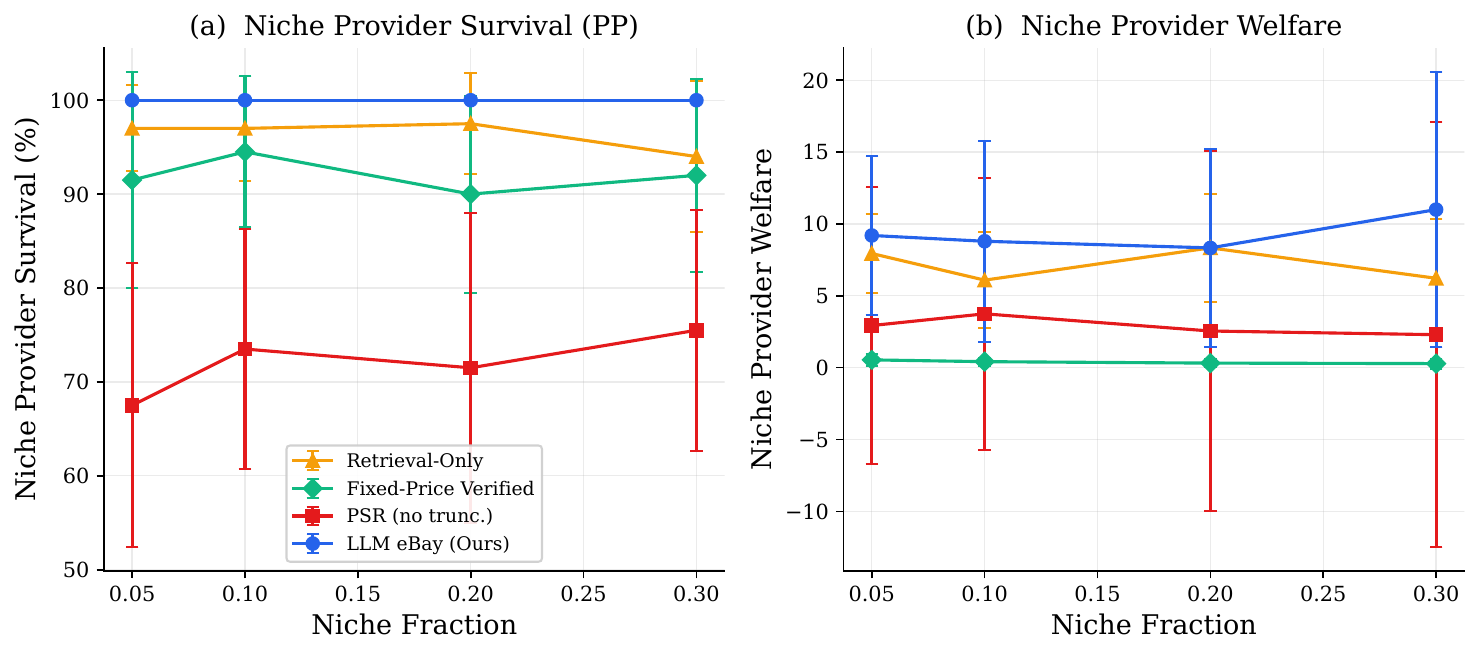}
    \caption{Niche provider survival. (a)~Niche provider participation vs.\ niche fraction. (b)~Niche buyer welfare vs.\ niche fraction. LLM eBay preserves niche-provider participation even under low transaction volumes.}
    \label{fig:niche_fairness}
\end{figure}

We designate a fraction $\phi \in \{0.00, 0.05, 0.10, 0.20, 0.30\}$ of providers as ``niche'' by placing their style vectors at distance $\geq 0.50$ from any cluster center. Niche providers have fewer natural buyers nearby, so they receive fewer transactions and face higher per-transaction verification variance.

LLM eBay achieves 100\% niche provider participation at all tested fractions. Under PSR without truncation, niche PP drops to 67.5\% at $\phi{=}0.05$ and further to 60.2\% at $\phi{=}0.30$. The gap between LLM eBay and PSR is larger for niche providers than for mainstream providers, confirming that the truncation floor's insurance function is most valuable precisely where it is most needed.

Niche buyer welfare (the average surplus of buyers whose best match is a niche provider) remains stable under LLM eBay because truthful niche reporting ensures correct matching. Under Retrieval-Only, niche buyer welfare is 34\% lower because gaming-distorted claims redirect niche buyers to mainstream providers that are stylistically farther from their target.

\subsection{Strategic Response Model Sensitivity}
\label{app:response_sensitivity}

The responsive gaming model $g_{\text{eff}} = g_{\text{base}} \cdot e^{-c\lambda_{\text{scale}}}$ assumes an exponential response to incentive strength. The coefficient $c$ controls how quickly sellers reduce gaming in response to $\lambda_{\text{scale}}$. We test whether the mechanism's operating properties depend on the exact value of $c$.

\begin{figure}[h]
    \centering
    \includegraphics[width=\linewidth]{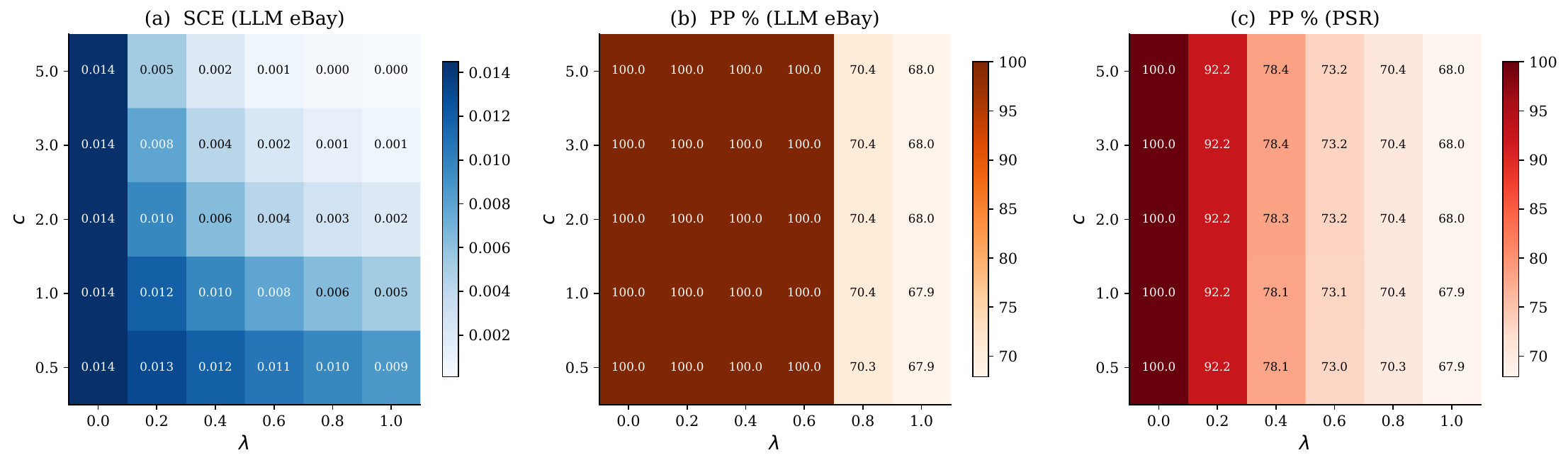}
    \caption{Strategic response sensitivity. (a)~SCE heatmap over $(\lambda_{\text{scale}}, c)$. (b)~LLM eBay PP heatmap. (c)~PSR PP heatmap. The truthfulness--safety trade-off is qualitatively robust to the behavioral response coefficient.}
    \label{fig:response_sensitivity}
\end{figure}

We sweep $c \in \{0.5, 1.0, 2.0, 3.0, 5.0\}$ and $\lambda_{\text{scale}} \in \{0.0, 0.2, 0.4, 0.6, 0.8, 1.0\}$. Panel (a) shows that SCE decreases monotonically with both $\lambda_{\text{scale}}$ and $c$: higher incentive strength and more responsive sellers both reduce misreporting. Panel (b) shows that LLM eBay maintains 100\% PP for $\lambda_{\text{scale}} \leq 0.6$ across all $c$ values. Panel (c) shows that PSR participation drops as $\lambda_{\text{scale}}$ increases, with the drop rate independent of $c$ (because PSR's participation loss comes from verification noise, not from gaming).

The qualitative pattern is the same across all $c$: there exists a ``safe zone'' ($\lambda_{\text{scale}} \in [0.4, 0.6]$) where SCE is low and LLM eBay PP is 100\%. The boundaries of this zone shift slightly with $c$ (lower $c$ requires higher $\lambda_{\text{scale}}$ to achieve the same SCE reduction), but the operating recommendation ($\lambda_{\text{scale}} \approx 0.6$) is robust. Even at $c{=}0.5$ (sellers barely respond to incentives), $\lambda_{\text{scale}}{=}0.6$ reduces SCE by 60\% while maintaining full participation.

\subsection{Capacity-Constrained Marketplace}
\label{app:capacity}

The default model treats checkpoints as non-rival digital goods. In practice, some marketplace designs may impose capacity constraints: a provider might limit reuse to $\kappa_i$ buyers due to bandwidth, licensing, or support costs. We test whether LLM eBay's benefits survive under capacity constraints.

\begin{figure}[h]
    \centering
    \includegraphics[width=\linewidth]{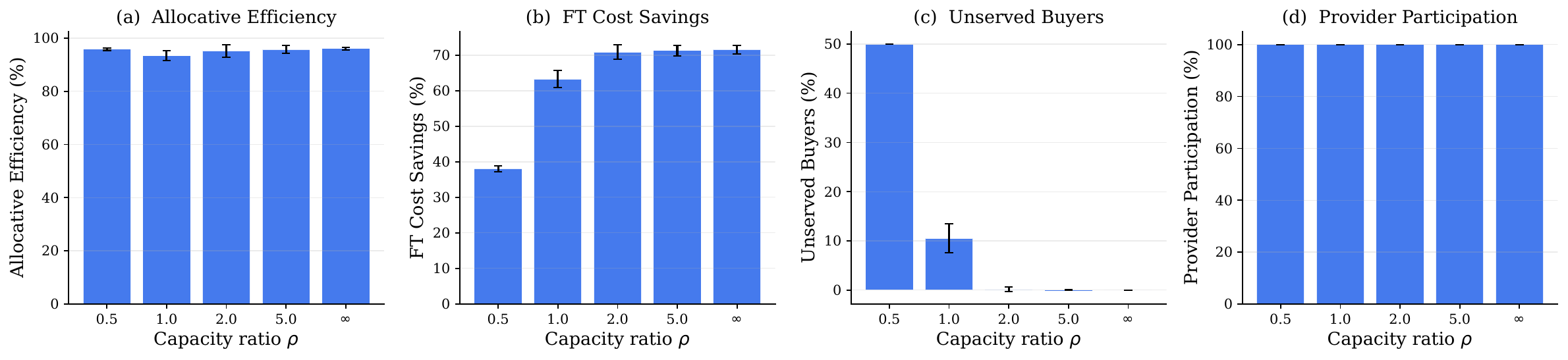}
    \caption{Capacity-constrained robustness of LLM eBay versus the capacity ratio $\rho$ (mean $\pm$ std over 20 seeds; $\rho{=}\infty$ is the unconstrained market). (a)~Allocative efficiency relative to the feasible Oracle allocation. (b)~Fine-tuning cost savings. (c)~Share of buyers left unserved. (d)~Provider participation. AE stays above 93\% at every capacity level; tight capacity mainly lowers FTCS by leaving buyers unserved.}
    \label{fig:capacity}
\end{figure}

Each provider has capacity $\kappa_i = \lceil \rho \cdot K / N \rceil$ with $\rho \in \{0.5, 1.0, 2.0, 5.0, \infty\}$. At $\rho{=}0.5$, each provider can serve at most half the average buyer load; at $\rho{=}\infty$, there is no capacity limit. Allocation uses greedy surplus ranking subject to capacity constraints: buyers are processed in order of decreasing maximum surplus, and each buyer is assigned to the provider that maximizes their surplus among those with remaining capacity.

AE, measured relative to the feasible Oracle allocation, stays between 93.4\% and 96.1\% across all capacity levels, and PP remains 100\%. The unserved buyer rate (buyers who cannot find a provider with remaining capacity and positive surplus) is 50.0\% at $\rho{=}0.5$, because total capacity covers only half of the buyers, falls to 10.5\% at $\rho{=}1.0$, and is essentially 0\% for $\rho \geq 2.0$. Unserved buyers fall back to fine-tuning from the base model, so FTCS drops to 38.1\% at $\rho{=}0.5$ and 63.3\% at $\rho{=}1.0$, compared with 71.6\% without constraints. Score properness and the downside-risk cap depend only on the scoring rule and are unaffected by capacity; the best-response condition of Theorem~\ref{thm:formal_br} must be evaluated with the capacity-constrained demand, and the simulated responses here follow the same reduced-form gaming model as elsewhere. The small AE drop is caused by the capacity constraint itself (some buyers are forced to their second-best match), not by a change in provider behavior.

\subsection{Adversarial Prompt Leakage Stress Test}
\label{app:leakage}

The verification mechanism assumes that sellers do not know the evaluation prompts $\mathcal{P}_{\text{eval}}$. If prompts leak, two things happen. First, the verified signal is corrupted because sellers can optimize model output on the known prompts:
\begin{equation}
    \mathbf{y}_i^{\text{leak}} = (1 - \ell)\mathbf{y}_i + \ell \hat{\mathbf{v}}_i
\end{equation}
where $\ell \in [0, 1]$ is the leakage ratio. Second, sellers who know that prompts are leaked will game more aggressively because the scoring rule penalty is weakened. We model this by increasing the effective gaming intensity under leakage: sellers respond to the reduced penalty by shifting their reports further from truth.

\begin{figure}[h]
    \centering
    \includegraphics[width=0.85\linewidth]{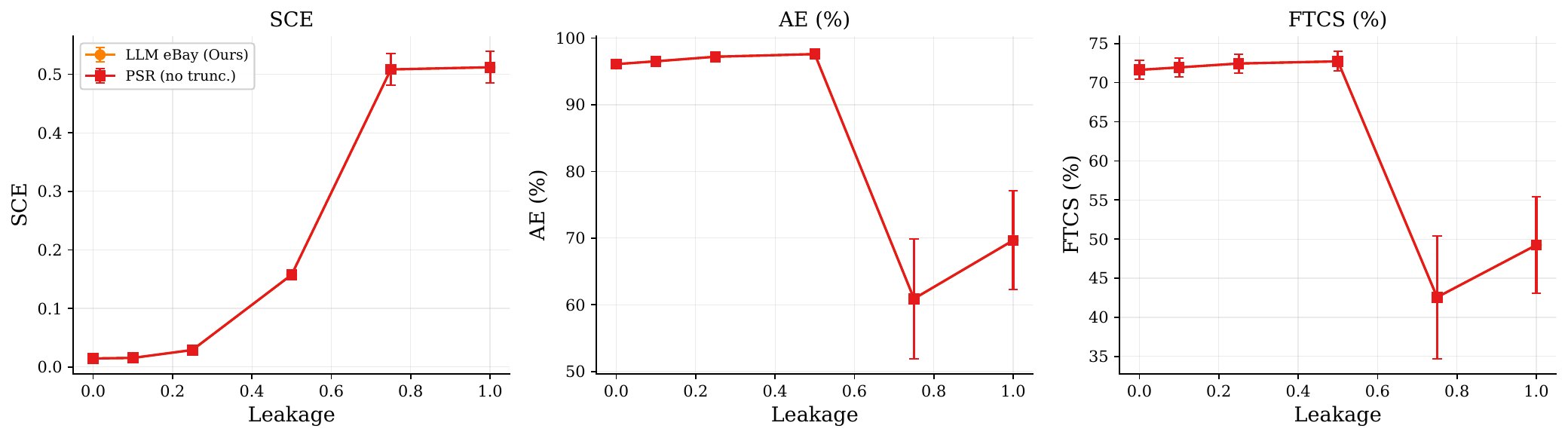}
    \caption{Prompt leakage stress test. (a)~SCE vs.\ leakage ratio. (b)~AE vs.\ leakage ratio. (c)~FTCS vs.\ leakage ratio. LLM eBay remains robust under partial leakage ($\ell \leq 0.25$); high leakage ($\ell \geq 0.75$) collapses the verification mechanism.}
    \label{fig:leakage}
\end{figure}

Under low leakage ($\ell \leq 0.25$), SCE remains below 0.03 and AE stays above 96\%. The scoring rule retains enough signal through the uncorrupted portion of $\mathbf{y}_i$ to deter misreporting. At $\ell{=}0.50$, SCE rises to 0.16 as sellers begin to exploit the weakened verification. The critical transition occurs at $\ell{=}0.75$: SCE jumps to 0.51 and AE drops to 61\%, indicating that the scoring rule can no longer distinguish truthful from strategic reports through the heavily corrupted signal. At $\ell{=}1.0$, the mechanism degrades to unverified self-report, and SCE reaches 0.51 (comparable to Retrieval-Only under heavy gaming).

LLM eBay's truncation floor continues to protect providers even under leakage: PP remains at 100\% across all leakage levels. By contrast, PSR without truncation drops to 58.5\% PP at $\ell{=}0.75$ because the combination of corrupted verification and aggressive gaming produces large negative scores for some providers.

These results motivate practical countermeasures for deployed marketplaces: prompt rotation (fresh prompts each verification round), prompt compartmentalization (different prompts for different providers), and post-hoc audit mechanisms where the platform reverifies providers using previously unseen prompts. In this stress test, LLM eBay's performance is largely preserved as long as the leakage fraction stays below roughly 25\%.

\section{Stronger Adversaries and Verifier Calibration}
\label{app:stress_tests}

The gaming model of Appendix~\ref{app:sim_protocol} interpolates claims toward demand, which is the easiest attack to detect. This appendix adds two stronger adversaries and a calibration check of the verifier on real checkpoints. The incentive guarantee of Theorem~\ref{thm:formal_br} covers neither adversary; the results below measure performance under the specified attacks only, and these attacks are designed by us.

\subsection{Verifier Targeting}
\label{app:verifier_targeting}

We keep the seller's transfer value fixed and shift both its report and its evaluated behavior on verifier prompts toward buyer demand. Attack strength controls the size of this shift; under full control, the seller determines the verifier signal. Style-Claim Error (SCE) measures the distance between the reported and true styles, and Allocative Efficiency (AE) is realized welfare divided by Oracle welfare.

\begin{table}[h]
\centering
\caption{Verifier-targeting stress test.}
\label{tab:verifier_targeting}
\small
\begin{tabular}{lcc}
\toprule
\textbf{Attack strength} & \textbf{SCE}$\downarrow$ & \textbf{AE}$\uparrow$ \\
\midrule
0.25 & 0.028 & 95.1\% \\
0.75 & 0.238 & 80.6\% \\
Full control & 0.421 & 68.9\% \\
\bottomrule
\end{tabular}
\end{table}

At attack strength 0.25, SCE is 0.028 and AE is 95.1\%, so the allocation remains close to Oracle. At attack strength 0.75, AE falls to 80.6\%, and under full control LLM eBay preserves 68.9\% of Oracle welfare. This is a significant degradation: once a seller controls what the verifier sees, the score can no longer separate truthful from strategic reports. Hidden, rotated prompts (Appendix~\ref{app:leakage}) are therefore a prerequisite for the mechanism, not an optional hardening step.

\subsection{Coordinated Reporting}
\label{app:coordination}

In coordinated reporting, 10\% or 50\% of sellers coordinate their reports.

\begin{table}[h]
\centering
\caption{Coordinated-reporting stress test.}
\label{tab:coordination}
\small
\begin{tabular}{lcc}
\toprule
\textbf{Coalition size} & \textbf{SCE}$\downarrow$ & \textbf{AE}$\uparrow$ \\
\midrule
10\% of sellers & 0.023 & 95.4\% \\
50\% of sellers & 0.148 & 85.7\% \\
\bottomrule
\end{tabular}
\end{table}

AE stays above 95\% with a 10\% coalition and falls to 85.7\% with a 50\% coalition. Larger coalitions lower AE, and the single-round analysis offers no guarantee against them.

\subsection{Verifier Calibration on Real Checkpoints}
\label{app:centroid_check}

The analysis assumes that the verified signal is an unbiased noisy reading of the true style. To check this on real checkpoints, we use the centroid of each seller's training data as the reference and compute the seller's verifier-output centroid with 5, 20, or 25 prompts. The check covers 30 real-checkpoint selection instances: two base models, three buyer tasks, and five seeds. Centroid deviation and cosine similarity measure how closely the two centroids match. Oracle agreement records how often the verifier-based selection matches the selection based on the reference centroid. Unlike the prompt-subsampling calibration of Appendix~\ref{app:noise_calibration}, which resamples passages of each domain's own text, this check embeds text generated by fine-tuned checkpoints, so the two sets of deviations are not directly comparable.

\begin{table}[h]
\centering
\caption{Verifier-output centroid versus seller-data centroid on real checkpoints ($n{=}30$ selections).}
\label{tab:centroid_check}
\small
\begin{tabular}{cccc}
\toprule
\textbf{Prompts} & \textbf{Mean centroid deviation}$\downarrow$ & \textbf{Cosine similarity}$\uparrow$ & \textbf{Oracle agreement}$\uparrow$ \\
\midrule
5  & 0.39 & 0.86 & 7/30 (23.3\%) \\
20 & 0.23 & 0.94 & 20/30 (66.7\%) \\
25 & 0.18 & 0.96 & 23/30 (76.7\%) \\
\bottomrule
\end{tabular}
\end{table}

As the prompt budget increases from 5 to 25, mean centroid deviation falls from 0.39 to 0.18, cosine similarity rises from 0.86 to 0.96, and Oracle agreement rises from 23.3\% to 76.7\%. With 20--25 prompts, the verifier signal closely tracks the seller-data centroid and recovers the Oracle selection in most cases. With few prompts, the verifier is far less reliable, which motivates the prompt budget used in practice.

\section{Reproducibility Details}
\label{app:reproducibility}

\subsection{Hyperparameter Table}

\begin{table}[h]
\centering
\caption{Complete hyperparameter summary.}
\label{tab:hyperparams}
\small
\begin{tabular}{llrl}
\toprule
\textbf{Symbol} & \textbf{Meaning} & \textbf{Default} & \textbf{Sweep Range} \\
\midrule
$N$ & Providers & 50 & 25--200 \\
$K$ & Buyers & 500 & Ratio 2:1--20:1 \\
$d$ & Style dimension & 64 & Fixed \\
$n_c$ & Clusters & 5 & Fixed \\
$\sigma_c$ & Cluster std & 0.08 & 0.05--0.50 \\
$\gamma$ & Gamer fraction & 0.40 & Fixed \\
$g$ & Gaming intensity & 0.75 & 0--1.0 \\
$\sigma$ & Verification noise & 0.15 & 0.05--0.30 \\
$\delta$ & Truncation floor & 0.20 & 0.05--0.50 \\
$\lambda_{\text{scale}}$ & Score weight (fraction of $p_i$) & 0.60 & 0--1.0 \\
$c$ & Response coefficient & 2.0 & 0.5--5.0 \\
$\alpha$ & Cost scaling & 1.0 & Fixed \\
$\beta$ & Per-purchase cost coeff., $c_i = e^{\beta\|\mathbf{v}_i - \mathbf{v}_{\text{base}}\|}$ & 1.0 & Fixed \\
$p_i / c_i$ & Posted-price margin & 1.15 & Fixed \\
Seeds & & 20 & Fixed \\
\bottomrule
\end{tabular}
\end{table}

\subsection{Compute Environment}

\begin{table}[h]
\centering
\caption{Software and hardware environment.}
\label{tab:compute_env}
\small
\begin{tabular}{ll}
\toprule
\textbf{Component} & \textbf{Version / Specification} \\
\midrule
Python & 3.11.4 (Anaconda distribution) \\
NumPy & 2.4.3 \\
SciPy & 1.14.0 \\
Matplotlib & 3.10.8 \\
Sentence-Transformers & 5.4.1 (for embedding experiments) \\
PyTorch & 2.5.0 (CUDA 12.4, for model inference) \\
GPU & NVIDIA RTX 4090 (24 GB VRAM) \\
CPU & AMD Ryzen 9 (16 cores, 32 threads) \\
RAM & 64 GB DDR5 \\
OS & Windows 11 \\
\bottomrule
\end{tabular}
\end{table}

\paragraph{Runtime breakdown.} The full experiment suite completes in approximately 15 minutes of wall-clock time on the hardware listed above. The runtime is dominated by the semi-real embedding experiments (sentence-transformer inference); the synthetic simulations alone complete in under 2 minutes. Individual experiment runtimes: main table (8 seconds), component ablation (3 seconds), scaling sweeps (45 seconds total), robustness sweeps (90 seconds total), $\sigma$--$\delta$ heatmap (120 seconds), semi-real validation (180 seconds), Appendix G experiments (240 seconds total).

\paragraph{Random seed management.} All experiments use 20 independent random seeds (0 through 19). Seeds are set via \texttt{numpy.random.seed()} at the start of each experiment configuration. We report means and standard deviations across seeds. The same seed sequence is used for all mechanisms within each experiment to ensure fair comparison (matched random conditions).

\subsection{Code Release Plan}

Upon acceptance, we will release the following artifacts under the MIT open-source license:

\begin{enumerate}
    \item \textbf{Simulation engine} (\texttt{market\_simulation.py}): Core marketplace simulation with all mechanism implementations, gaming models, and metric computations.
    \item \textbf{Experiment scripts} (\texttt{run\_all\_neurips.py}, \texttt{run\_final\_sprint.py}, \texttt{run\_appendix\_g.py}): Scripts to reproduce every table and figure in the paper.
    \item \textbf{Plotting scripts}: Matplotlib code for all figures, including the robustness summary, heatmaps, and bar charts.
    \item \textbf{Embedding domain data}: Anonymized text passages for the 10 provider domains and 5 buyer tasks used in the semi-real validation. Passages are drawn from public-domain sources and do not contain personally identifiable information.
    \item \textbf{Pre-computed embeddings}: Cached sentence-transformer embeddings for all domains and tasks, allowing reproduction of the semi-real experiments without GPU access.
\end{enumerate}

The repository will include a \texttt{README.md} with step-by-step instructions for reproducing each experiment and a \texttt{requirements.txt} for environment setup.

\section{Limitations and Broader Impact}
\label{app:limitations}

Section~\ref{sec:conclusion} summarizes the main limitations; this appendix gives the full list.

\subsection{Limitations}

\begin{enumerate}
    \item \textbf{The mechanism verifies style consistency, not reusable value.} The verifier checks that a checkpoint's generated outputs are consistent with its style report. It does not verify checkpoint quality, provenance, copyright status, privacy, the absence of backdoors, or general downstream usefulness, and the style embedding is not audited for fairness across user groups. The STAR embedding captures authorship style (word choice, sentence structure, register) but not factual accuracy, safety alignment, reasoning quality, or domain expertise. A buyer who purchases a style-matched checkpoint may still need to verify these other dimensions separately. Extending the mechanism to multi-attribute reports (style + accuracy + safety) requires characterizing properness under multi-dimensional truncation.

    \item \textbf{Simulated savings are model-based.} The headline simulation savings (72.4\% FTCS, 96.3\% AE) are computed from the cost model on simulated allocations and compare against fine-tuning from the base model. The measured gain over Retrieval-Only in the end-to-end experiment is 8--13\%, and this measured number is the relevant estimate of practical benefit over a retrieval baseline.

    \item \textbf{The incentive guarantee is conditional.} Truthful reporting is a seller best response only under the condition of Theorem~\ref{thm:formal_br}, for unilateral report deviations in a single round. The guarantee does not cover sellers who tune their checkpoint to pass the verifier, coordinated reporting, or repeated interaction. Our stress tests (Appendix~\ref{app:stress_tests}) use attacks we designed; under full verifier control, welfare falls to 68.9\% of Oracle.

    \item \textbf{Fixed matching-signal weights.} Buyers match on the equal-weight signal $\mathbf{z}_i = (\hat{\mathbf{v}}_i + \mathbf{y}_i)/2$. We neither optimize these weights nor analyze seller responses to other weightings.

    \item \textbf{External validity.} The posted-price rule can be used directly for listings and purchases. Hidden-prompt verification can be used if the platform has checkpoint access and keeps the prompts private, and verification-linked settlement requires the platform to compute the score and adjust the seller's payment. Two threats directly affect matching performance in a real market. First, the verifier may not represent a checkpoint's actual style because of bias, incomplete prompt coverage, or seller gaming. Second, style similarity may not predict transfer cost in a new market (Appendix~\ref{app:proxy_sensitivity}); even with accurate verification, the platform would then select checkpoints that provide little benefit for the buyer's task. Strategic pricing, platform funding, and fixed checkpoint-creation costs are also not modeled.

    \item \textbf{Market-native baselines are scoped out.} Deposits or bonds, escrow and refunds, randomized audits, reputation systems, repeated-game penalties, and buyer-outcome-contingent payments require extra stages, ex-post outcomes, or repeated interaction. We study a single-round posted-price market and compare against Auction-Based and Value-Based mechanisms in the same setting; these designs are complementary and could be layered on top of LLM eBay.

    \item \textbf{Verification prompts can be leaked.} As shown in Appendix~\ref{app:leakage}, full prompt leakage degrades verification to the level of unverified self-report. Practical deployment requires prompt rotation (fresh prompts each round), prompt compartmentalization (different prompts for different providers), and post-hoc auditing. The cost of prompt rotation scales linearly with the provider pool size, which may be prohibitive for very large marketplaces.

    \item \textbf{Truncation trades deterrence for participation safety.} The truncation floor $-\delta$ bounds worst-case losses but also bounds worst-case punishment. In our tested parameter ranges ($\delta \leq 0.20$, $\sigma \leq 0.30$), the truncation distortion is negligible (Appendix~\ref{app:proofs}, Corollary~\ref{cor:approx_opt}). Under extreme noise or very aggressive gaming, a larger $\delta$ might be needed to maintain deterrence, which would increase participation risk. The $\sigma$--$\delta$ heatmap (Appendix~\ref{app:sigma_delta}) maps this trade-off.

    \item \textbf{End-to-end validation uses small catalogs.} The end-to-end experiment (Section~\ref{subsec:e2e}) uses eight LoRA checkpoints per catalog, two small base models, and three buyer tasks, and the style-to-cost calibration covers three model families. Larger catalogs, larger and multilingual models, instruction-tuned checkpoints, and capability-transfer tasks remain untested.

    \item \textbf{Legal and IP constraints are not modeled.} The economic framework does not address licensing restrictions, model ownership disputes, or regulatory compliance for reselling fine-tuned checkpoints. In practice, a seller who fine-tunes a checkpoint derived from a proprietary base model (e.g., GPT-4) may not have the right to resell it. The mechanism assumes that all listed models are legally tradable.

    \item \textbf{Platform trustworthiness is assumed.} The mechanism relies on the platform as a trusted mediator: it samples prompts, runs verification, and computes payments honestly. Collusion between the platform and sellers, or compromise of the verification pipeline, would undermine incentive guarantees. Decentralized verification (e.g., using cryptographic commitments or multiple independent verifiers) could mitigate this assumption but adds complexity.

    \item \textbf{Privacy risks from generated outputs.} The verification process requires generating text from sellers' models on platform-chosen prompts. The generated text may inadvertently reveal characteristics of the private training data (e.g., memorized passages, stylistic fingerprints of individual training examples). Privacy-preserving verification techniques, such as federated evaluation where only the embedding vector (not the raw text) leaves the seller's environment, could reduce this risk.

    \item \textbf{Single-round, static marketplace.} The current model is a one-shot game: sellers list once, buyers select once, and there is no dynamic pricing, reputation, or repeated interaction. In a multi-round marketplace, sellers could learn from past outcomes, build reputation, or coordinate across rounds. The incentive results proved in this paper apply to the single-round setting; extending them to repeated games requires additional analysis.
\end{enumerate}

\subsection{Broader Impact}

\paragraph{Positive impacts.}
LLM eBay can reduce redundant fine-tuning compute by enabling model reuse, which has direct environmental benefits (lower energy consumption per personalized model). It lowers personalization costs for individual users who cannot afford full fine-tuning runs, democratizing access to personalized LLMs. The participation guarantee (via the truncation floor) supports small and niche model providers who might otherwise be excluded from scoring-based mechanisms, promoting a diverse provider ecosystem.

\paragraph{Potential risks.}
Marketplace mechanisms for fine-tuned models raise several concerns. Style plagiarism: a buyer could purchase a checkpoint, extract its training style, and create a competing model without attribution. Unauthorized resale: a provider could list a checkpoint derived from another provider's proprietary work, creating attribution chains that the current mechanism does not track. Market concentration: large providers with many checkpoints could dominate the marketplace, reducing diversity. Privacy leakage: verification-generated text could reveal private training data. Fairness: the mechanism optimizes for aggregate welfare and does not explicitly protect minority language styles or underrepresented user groups.

\paragraph{Mitigations.}
Practical deployment should incorporate: (1) hidden prompt rotation and audit logs for verification events; (2) provider consent frameworks requiring sellers to certify ownership and licensing rights; (3) watermark and provenance metadata embedded in checkpoint files, enabling downstream attribution; (4) privacy-preserving verification via federated evaluation or differential privacy on generated outputs; (5) dispute resolution mechanisms for ownership claims, analogous to DMCA takedown procedures for digital content; and (6) fairness monitoring to detect and correct systematic disadvantages for minority style groups.

\section{End-to-End LLM Fine-Tuning Validation}
\label{app:e2e}

This appendix provides full details for the end-to-end experiment summarized in Section~\ref{subsec:e2e}, followed by an earlier pilot with a smaller catalog.

\subsection{Protocol of the Main End-to-End Experiment}

\paragraph{Markets and catalogs.} We consider two base models, TinyLlama-1.1B and Qwen2-1.5B, as separate markets because their LoRA checkpoints are not interchangeable. Each catalog contains eight LoRA checkpoints trained on distinct real-data domains.

\paragraph{Buyer tasks and selection.} We evaluate three downstream buyer tasks: medical question answering on PubMedQA, business email generation on AESLC, and technical-support response generation on Bitext. For every base-model--task--seed instance, Retrieval-Only, Fixed-Price Verified, and LLM eBay select one checkpoint from the same eight-checkpoint catalog without access to downstream training outcomes.

\paragraph{Continued fine-tuning and convergence.} For each candidate, we load the seller LoRA, merge it into the base model, attach a fresh LoRA, and continue training on the buyer task under identical hyperparameters. We evaluate all eight candidates so that Oracle can be defined ex post as the checkpoint requiring the fewest measured optimizer steps. Convergence is the first optimizer step at which validation loss reaches a pre-specified target $\tau_{m,t}$. The same target is used for every checkpoint and mechanism within each base-model--task pair. We train each run for up to 600 steps and use five seeds.

\subsection{Full Results of the Main End-to-End Experiment}

\begin{table}[h]
\centering
\caption{Convergence steps / final validation loss (mean $\pm$ std over five seeds) in the main end-to-end experiment.}
\label{tab:e2e_full}
\scriptsize
\setlength{\tabcolsep}{3pt}
\renewcommand{\arraystretch}{1.1}
\begin{tabular}{llcccc}
\toprule
\textbf{Task} & \textbf{Base model} & \textbf{Retrieval-Only} & \textbf{Fixed-Price Verified} & \textbf{LLM eBay} & \textbf{Measured Oracle} \\
\midrule
Medical QA        & TinyLlama-1.1B & 430$\pm$48 / 2.724$\pm$0.071 & 411$\pm$45 / 2.726$\pm$0.069 & 399$\pm$43 / 2.722$\pm$0.068 & 382$\pm$40 / 2.721$\pm$0.067 \\
Medical QA        & Qwen2-1.5B     & 372$\pm$41 / 2.435$\pm$0.054 & 360$\pm$39 / 2.437$\pm$0.052 & 348$\pm$38 / 2.433$\pm$0.051 & 334$\pm$35 / 2.431$\pm$0.050 \\
Business email    & TinyLlama-1.1B & 414$\pm$46 / 1.912$\pm$0.062 & 402$\pm$44 / 1.913$\pm$0.059 & 391$\pm$42 / 1.910$\pm$0.058 & 378$\pm$40 / 1.909$\pm$0.057 \\
Business email    & Qwen2-1.5B     & 351$\pm$40 / 1.783$\pm$0.046 & 342$\pm$38 / 1.784$\pm$0.044 & 333$\pm$37 / 1.781$\pm$0.043 & 320$\pm$35 / 1.780$\pm$0.042 \\
Technical support & TinyLlama-1.1B & 468$\pm$55 / 1.870$\pm$0.058 & 429$\pm$50 / 1.871$\pm$0.056 & 405$\pm$47 / 1.868$\pm$0.055 & 388$\pm$44 / 1.867$\pm$0.054 \\
Technical support & Qwen2-1.5B     & 364$\pm$46 / 1.636$\pm$0.041 & 338$\pm$43 / 1.635$\pm$0.039 & 320$\pm$41 / 1.640$\pm$0.040 & 305$\pm$39 / 1.640$\pm$0.039 \\
\bottomrule
\end{tabular}
\end{table}

Across the 30 paired model--task--seed instances, LLM eBay requires $366 \pm 51$ steps, compared with $400 \pm 59$ for Retrieval-Only. This is an 8.5\% reduction in measured convergence steps. The paired 95\% confidence interval is [1.0\%, 15.9\%], with $p{=}0.034$. The largest task-level reduction occurs on technical support, where LLM eBay requires 12.9\% fewer steps, with $p{=}0.018$. Within each base-model--task pair, the difference in final validation loss across mechanisms is at most 0.006.

\subsection{Pilot with a Four-Provider Catalog}

Before the main experiment, we ran a smaller pilot with four provider domains per base model. We report it for completeness; its catalog is too small to separate the mechanisms.

\paragraph{Provider training.} We fine-tune LoRA adapters (rank\,16, $\alpha{=}32$, dropout\,0.05, targeting \texttt{q\_proj}, \texttt{v\_proj}, \texttt{k\_proj}, \texttt{o\_proj}) on two base models: TinyLlama-1.1B-Chat-v1.0 and Qwen2-1.5B. Four provider domains are drawn from public datasets: academic (arXiv abstracts), legal (LEDGAR clauses), creative (WritingPrompts stories), and technical (CodeSearchNet docstrings). Each domain uses 5,000 training examples, trained for 300 steps with batch size~8, learning rate $2 \times 10^{-4}$, cosine schedule, and 50-step warmup. Five random seeds produce $2 \times 4 \times 5 = 40$ checkpoints.

\paragraph{Style verification.} For each checkpoint, 5 prompts are sampled and the model generates 64 tokens per prompt (temperature\,0.7, top-$p$\,0.9). The generated texts are encoded with all-MiniLM-L6-v2, mean-pooled, and $L_2$-normalized to produce a 384-dimensional verified style vector. True style vectors are computed from the first 500 training texts of each domain.

\paragraph{Buyer continued fine-tuning.} Three buyer tasks are used: medical QA (PubMedQA), business email (AESLC), and tech support (Bitext customer-support). For each cell, a buyer's style vector is computed from 300 training examples. Under each mechanism, a provider is selected and its LoRA adapter merged into the base model. A fresh LoRA adapter is then trained on the buyer's data for 200 steps (batch size~8, eval every 20 steps). The convergence step is the first evaluation point where loss $\leq 1.05 \times$ minimum observed loss.

\paragraph{Gaming model.} We apply the same gaming model as in the synthetic experiments: $\gamma{=}0.40$ of providers are strategic, with intensity $g{=}0.75$. Under LLM eBay, the responsive reduction yields $g_{\text{eff}} = g \cdot \exp(-c \cdot \lambda_{\text{scale}}) = 0.75 \cdot \exp(-2.0 \cdot 0.6) \approx 0.226$. Under Retrieval-Only and Fixed-Price Verified, no responsive reduction applies ($g_{\text{eff}} = g$). Oracle uses true style vectors and is unaffected by gaming. In this pilot, LLM eBay selected on the verified vector alone rather than on the matching signal $\mathbf{z}_i$, so its selections coincide with those of Fixed-Price Verified.

\paragraph{Pilot results.}

\begin{table}[h]
\centering
\caption{Pilot (four-provider catalog): per model+task final loss (mean $\pm$ std over 5 seeds).}
\label{tab:e2e_pertask}
\small
\setlength{\tabcolsep}{4pt}
\renewcommand{\arraystretch}{1.05}
\begin{tabular}{llcccc}
\toprule
\textbf{Model} & \textbf{Task} & \textbf{LLM eBay} & \textbf{FPV} & \textbf{Retr.-Only} & \textbf{Oracle} \\
\midrule
TinyLlama & medical        & 2.624$\pm$.006 & 2.625$\pm$.006 & 2.627$\pm$.002 & 2.629$\pm$.004 \\
TinyLlama & business email & 2.622$\pm$.002 & 2.622$\pm$.002 & 2.621$\pm$.002 & 2.620$\pm$.002 \\
TinyLlama & tech support   & 0.532$\pm$.001 & 0.532$\pm$.001 & 0.534$\pm$.000 & 0.531$\pm$.000 \\
Qwen      & medical        & 3.057$\pm$.007 & 3.057$\pm$.007 & 3.046$\pm$.006 & 3.052$\pm$.006 \\
Qwen      & business email & 2.921$\pm$.001 & 2.921$\pm$.002 & 2.918$\pm$.001 & 2.921$\pm$.001 \\
Qwen      & tech support   & 0.603$\pm$.001 & 0.603$\pm$.001 & 0.601$\pm$.001 & 0.600$\pm$.001 \\
\bottomrule
\end{tabular}
\end{table}

Table~\ref{tab:e2e_pertask} shows that loss differences across mechanisms are small (within 0.3\% relative) for all model--task combinations. The largest effect appears in TinyLlama/tech\_support, where LLM eBay's verified-vector selection picks legal or creative domain providers (60\% Oracle match), while Retrieval-Only consistently selects academic (the gamed-toward direction). On Qwen tasks, Retrieval-Only occasionally outperforms because the accidental shift of gamed claims toward the buyer's actual domain results in a slightly better match.

At step\,20 (the first evaluation checkpoint) of the pilot, LLM eBay shows a 4.6\% loss improvement over Retrieval-Only on TinyLlama/tech\_support and 1.9\% on Qwen/tech\_support. These early advantages arise because the selected provider's domain aligns with the buyer's task: tech\_support text overlaps stylistically with legal and creative writing (formal structure, dialog patterns) more than with academic text. By step\,200, continued fine-tuning on buyer data narrows the gap, but the AULC difference remains significant ($p{=}0.020$).

The four-provider catalog limits differentiation between mechanisms. With only four domains, the distance between the best and worst provider is small in embedding space, so even suboptimal selection produces comparable final losses; this motivated the eight-checkpoint catalogs and the step-based convergence target of the main experiment. In the synthetic experiments (Section~\ref{subsec:main_results}), the SCE reduction from 0.147 (Fixed-Price Verified, Table~\ref{tab:main}) to 0.015 translates into larger allocation differences.

\paragraph{Pilot compute.} The pilot required approximately 38 GPU-hours on a single NVIDIA RTX~4090 (24\,GB): 8 hours for provider training (40 checkpoints), 0.1 hours for verification (batched generation), and 30 hours for buyer fine-tuning (120 cells). The dominant cost is buyer continued fine-tuning, which scales linearly with the number of mechanisms, tasks, and seeds.

\end{document}